\documentclass[manuscript, screen]{acmart} 

\usepackage{xcolor, colortbl}
\usepackage{hyperref}
\usepackage{graphicx}
\usepackage{float}

\usepackage{enumitem}

\makeatletter
\let\MYcaption\@makecaption
\makeatother
\usepackage[font=footnotesize]{subcaption}
\makeatletter
\let\@makecaption\MYcaption
\makeatother

\usepackage{amsmath}
\usepackage{amsthm}
\allowdisplaybreaks[1] 

\usepackage{amssymb}
\usepackage{xfrac}
\usepackage{mathtools}

\newcommand{\ip}[2]{\left\langle{#1},{#2}\right\rangle}

\usepackage{multirow}
\usepackage{fancyhdr}
\usepackage[ruled, linesnumbered, noend]{algorithm2e}

\newtheorem{theorem}{Theorem}[section]
\newtheorem{corollary}{Corollary}[theorem]
\newtheorem{lemma}[theorem]{Lemma}
\newtheorem{example}{Example}[section]

\usepackage{fancyvrb}
\usepackage{diagbox}
\newcommand{\ourdiagbox}[2]{\diagbox[dir=SW, linecolor=lightgray, innerrightsep=2pt, innerleftsep=2pt]{#1}{#2}}
\usepackage{adjustbox}
\usepackage{array}
\newcolumntype{x}[1]{>{\centering\arraybackslash\hspace{0pt}}p{#1}}
\newcolumntype{R}[2]{%
    >{\adjustbox{angle=#1,lap=\width-(#2)}\bgroup}%
    l%
    <{\egroup}%
}

\graphicspath{ {images/} }

\title{Exploiting Behavioral Side-Channels in Observation Resilient Cognitive Authentication Schemes}

\author{Benjamin Zi Hao Zhao}
\affiliation{%
  \institution{University of New South Wales and Data61, CSIRO}
  \city{Sydney}
  \country{Australia}
}
\email{benjamin.zhao@unsw.edu.au}

\author{Hassan Jameel Asghar}
\affiliation{%
  \institution{Macquarie University}
  \city{Sydney}
  \country{Australia}
}
\email{hassan.asghar@mq.edu.au}

\author{Mohamed Ali Kaafar}
\affiliation{%
  \institution{Macquarie University}
  \city{Sydney}
  \country{Australia}
}
\email{dali.kaafar@mq.edu.au}

\author{Francesca Trevisan}
\affiliation{%
  \institution{University of Surrey}
  \city{Guildford}
  \country{United Kingdom}
}
\email{f.trevisan@surrey.ac.uk}

\author{Haiyue Yuan}
\affiliation{%
  \institution{University of Surrey}
  \city{Guildford}
  \country{United Kingdom}
}
\email{haiyue.yuan@surrey.ac.uk}

\renewcommand{\shortauthors}{Zhao, et al.}
\renewcommand{\shorttitle}{Exploiting Behavioral Side-Channels in ORAS}

\begin{abstract}
Observation Resilient Authentication Schemes (ORAS) are a class of shared secret challenge-response identification schemes where a user mentally computes the response via a cognitive function to authenticate herself such that eavesdroppers cannot readily extract the secret. Security evaluation of ORAS generally involves quantifying information leaked via observed challenge-response pairs. However, little work has evaluated information leaked via human behavior while interacting with these schemes. A common way to achieve observation resilience is by including a modulus operation in the cognitive function. This minimizes the information leaked about the secret due to the many-to-one map from the set of possible secrets to a given response. In this work, we show that user behavior can be used as a side-channel to obtain the secret in such ORAS. Specifically, the user's eye-movement patterns and associated timing information can deduce whether a modulus operation was performed (a fundamental design element), to leak information about the secret. We further show that the secret can still be retrieved if the deduction is erroneous, a more likely case in practice. We treat the vulnerability analytically, and propose a generic attack algorithm that iteratively obtains the secret despite the ``faulty'' modulus information. We demonstrate the attack on five ORAS, and show that the secret can be retrieved with considerably less challenge-response pairs than non-side-channel attacks (e.g., algebraic/statistical attacks). In particular, our attack is applicable on Mod10, a one-time-pad based scheme, for which no non-side-channel attack exists. We field test our attack with a small-scale eye-tracking user study.
\end{abstract}



\begin{document}
\maketitle

\section{Introduction}
A longstanding issue with the prevailing methods of authenticating users via passwords and PINs is their vulnerability to observation. The user secrets (password or PIN) is entirely compromised after a single observation via, for instance, shoulder-surfing or a hidden camera. A growing number of reported incidents indicate that this is not just a theoretical vulnerability \cite{KrebsonSecurity},
prompting widespread proposals for alternative authentication schemes. These include biometric authentication (fingerprint, iris, etc.) and one-time passwords, either as standalone systems or in a multi-factor configuration alongside passwords. 
Another alternative is observation resilient challenge-response authentication schemes that rely on human cognition. In such schemes, the verifier (service provider) prompts the user to prove possession of a shared secret, through a series of challenges to whom the user has to respond to by (mentally) computing some cognitive function. The cognitive function is designed in a way that an eavesdropping adversary needs to observe multiple challenge-response pairs to retrieve the secret. We call these schemes observation resilient authentication schemes (ORAS).

An example of such schemes is the \textit{Mod10} scheme~\cite{wilfong1999method}. The user has a 4-digit PIN as the secret. The challenge consists of a random 4-digit number (communicated through a covert channel). The user computes the modulo 10 sum of each of the four digits in the secret with the corresponding digits in the challenge, and submits the 4-digit response. 
The use of the modulus operation is a common design element in many ORAS (e.g. ~\cite{hopper2001secure, li-shum, asghar2013does, chauhan2017behaviocog, matsumoto-attempt, kelleyimpact, asghar-acns}), as it makes them resilient to observation by reducing information leakage. For instance, the dot product of a random binary secret vector with a public binary vector, i.e., the challenge, leaks more information about the secret compared to when the dot product is reduced modulo 2 (only revealing its parity).

An interesting subcategory of ORAS which often employs the modulus operation is the so called $k$-out-of-$n$ ORAS. In these schemes, the secret is a mutually agreed upon set of items (between the user and the service) of cardinality $k$, selected from a larger pool of $n$ items. A challenge contains a random subset of the $n$ items, which is displayed on a device carried by the user. The cognitive function requires, at least, the identification of any of the $k$ secret items. The device itself does not store the user secret, and simply serves as an intermediary, relaying messages. Different realizations of these schemes exist based on how the cognitive function is constructed (which should be easy enough for the user to perform mentally). Examples of such schemes include the Hopper and Blum (HB) scheme~\cite{hopper2001secure}, the (modified) FoxTail scheme (FT)~\cite{asghar2013does}, and BehavioCog (BC)~\cite{chauhan2017behaviocog}, among others (see Section \ref{sec:existing_protocols} for scheme descriptions). An advantage of these schemes is that their security can be quantitatively analyzed by studying the mathematical properties of the cognitive function, and the information leaked through challenge-response pairs.

The security analysis accompanying the proposals of almost all ORAS, including $k$-out-of-$n$ variants, only considers ``flat transcripts'' of challenges-response pairs, ignoring entirely the interaction of the user with the ``relay'' device during computation of the cognitive function. Observing human behavior while interacting with the device is likely to reveal more information about the secret, e.g., if the user dwells over a particular spot on the device's screen. These issues have been raised before \cite{wiedenbeck2006design}\cite{leblanc2010guessing};
however there is no quantitative analysis of how such human behavior can be exploited to compromise the secret, barring some work on timing attacks which exploits the variation in time taken by humans when responding to challenges~\cite{vcagalj2015timing}. In this paper, we analyze how information obtained from user behavior while processing challenges in a wide class of ORAS (one that employs a modulus operation) can compromise the user's secret. In particular, we consider an adversary which can not only observe challenge-response transcripts but also user's eye movements with varying accuracy. 
This information could be obtained through pinhole cameras like those found on ATMs \cite{KrebsonSecurity}, and does not require an adversary to have control of the device's camera.
We show how this adversary can launch an attack on these schemes to obtain the user secret after observing far fewer authentication rounds (number of challenge-response pairs) than attacks which only consider challenge-response transcripts. 

In more detail, our main contributions are as follows: 
\begin{itemize}
\item We analyze a wide class of ORAS in which the cognitive function involves a modulus operation. By using a generic $k$-out-of-$n$ ORAS, we show in Section \ref{sec:mod_bias} that certain responses are more likely a result of a modulus or a non-modulus operation.\footnote{e.g., consider the sum of two integers modulo 10. The sum $5 + 6$ requires a modulus operation, whereas $5 + 3$ does not.} Furthermore, we show that knowing whether a modulus operation was performed or not in a given challenge can leak information that can lead to quicker retrieval of the secret.

\item We propose an algorithm to obtain the user's secret using possibly faulty information about whether the \emph{modulus event} has occurred or not in three proposed $k$-out-of-$n$ ORAS from the literature (BehavioCog, FoxTail and HB). By simulating varying degrees of information accuracy about the modulus event, we show that the resulting attack retrieves the user secret in far fewer authentication rounds (challenge-response pairs) than (efficient) non-side-channel attacks, e.g., Gaussian elimination. For instance, even with an imbalanced simulated accuracy of 1.0 in detecting the modulus, and 0.6 in detecting non-modulus events, we can find the user's secret in 280 rounds for BehavioCog, 390 for FoxTail, and 909 for HB. This reduces to 474 rounds for BehavioCog, 666 for Foxtail, and 1555 for HB, when the latter is reduced to 0.35. In comparison, efficient algebraic attacks on these schemes (Gaussian elimination) require 900 rounds for BehavioCog~\cite{chauhan2017behaviocog} and 16,290 rounds for FoxTail~\cite{asghar2013does}, whereas the HB scheme does not have any efficient attack and hence no bound on the number of rounds. These results are shown in Section \ref{sec:algorithms}.

\item
In Section \ref{sec:classifier}, we perform a small-scale eye-tracking user study with 11 users on BehavioCog, as a field test to evaluate how a user's eye-movement behavior during challenges can potentially expose information about the secret by indicating a modulus/no-modulus event. We identify and derive behavioral features from the eye-movement side-channel, e.g. total challenge time, and duration of last fixation. Using these features we train classifiers to predict the modulus and no-modulus events. 
We use leave-one-user-out verification to demonstrate event-specific behavioral information independent of users, and to avoid over-fitting user-specific behaviors.

\item Continued in Section \ref{sec:classifier}, we demonstrate real-world attack feasibility on $k$-out-of-$n$ ORAS by considering adversaries with varying technological capabilities. Four adversarial levels are considered linked to the detail of information available; from the coarsest---only timing information, to the finest---timing information dwelling on a specific item. 
This information is obtainable with access to a camera directed at the user's face \cite{krafka2016eye}, a likely scenario with covert pinhole cameras already found in instances of ATM skimming \cite{KrebsonSecurity}. 
In comparison to the aforementioned efficient algebraic attacks, we can deduce user secrets in 435, 589, and 1,346 rounds in BehavioCog, FoxTail and HB, respectively.\footnote{These numbers of required rounds (reported in Section~\ref{sec:classifier}) are from accuracy levels obtained via the user study. This is in contrast to the rounds required from simulated accuracy levels (reported in Section~\ref{sec:algorithms}) as mentioned before. See Section~\ref{sec:mod_classifier} for the reason behind the discrepancies in the reported number of rounds.}

\item Finally in Section \ref{sec:other_oras}, we demonstrate that our attacks are applicable to other ORAS as well (not just $k$-out-of-$n$ variants) as long as the cognitive function involves a modulus operation. Specifically, we evaluate the attack on PassGrids~\cite{kelleyimpact}, a locations and modulo arithmetic based scheme. 
We also implicate the Mod-10 scheme~\cite{wilfong1999method}, which is necessarily a one-time pad using a covert channel to communicate the pad. Being a one-time pad, the scheme is secure against an unlimited number of authentication rounds observed. However, we show that the secret can be compromised with knowledge of the modulus event (obtainable through user behavioral information). With a 10\% prediction error, it only takes an average of 36.1 rounds, to compromise a 4-digit PIN in Mod-10.
\end{itemize}

Compared to algebraic attacks (which only require passive observation of challenge-response pairs), behavioral side channels do require more effort from an attackers point of view. However, obtaining the resolution of user's eye movement information required in our attacks is not difficult given today's technology, and the attack in practice can be launched without much difficulty (e.g., by placing hidden cameras on frequently visited spots).
Furthermore, we show that even limited side-channel information such as time to respond to challenges is enough to retrieve the user's secret. This information can be obtained even without hidden cameras. Our attacks suggest that the design of ORAS should explicitly consider user behavior while executing the schemes, as a threat and source of information leakage about the secret especially since these schemes are purported to be observation resilient.

The rest of the paper is laid out as follows: Section \ref{sec:background} summarizes what an Observation Resilient Authentication Scheme is and how it functions. In Section \ref{sec:mod_bias}, modulus-related biases are mathematically analyzed. 
The proposal and simulation of algorithms to exploit faulty oracle information is presented in Section \ref{sec:algorithms}. A realization of the attack is performed with eye-movement side-channel information obtained from an eye-tracking user study, to field-test an attacker's capabilities in Section \ref{sec:classifier}. Finally our extension of the bias onto other authentication schemes is in Section \ref{sec:other_oras}.


\section{Background}
\label{sec:background}

\subsection{Observation Resilient Authentication Schemes}
A (human) authentication scheme is a shared secret challenge-response authentication scheme consisting of a setup phase and an authentication phase. In the setup phase a secret $S$ is shared between the prover (user) and a verifier (the authentication service). The authentication phase involves a series of challenges $c$ from the verifier (displayed on the user's device) and responses $r$ from the user, whereby the user mentally computes a public function $f$ of $c$ and $S$, returning the response $r$ to the verifier. We shall call each challenge together with its corresponding response as a challenge-response pair or a challenge-response round, interchangeably. After a specified number of challenge-response rounds, the verifier accepts the user if the responses are correct; otherwise the user is rejected. 

\paragraph*{Threat Model} We consider an eavesdropping adversary who can observe the interactions between the user and the server (during the authentication phase). Most prior work models this as giving the adversary one or more challenge-response pairs from the authentication phase. We extend this by also allowing the adversary to observe the interaction between the user and its device during the authentication phase (Figure~\ref{fig:process_diagram}). The \emph{transcript} of a challenge-response round is defined as this entire interaction: from challenge receipt, user interaction with the device during computation of $f$, to response submission. 

\paragraph*{Observation Resilience} A human authentication scheme is called observation resilient (ORAS) if no adversary (probabilistic polynomial time algorithm) can extract the secret with probability 1, after observing one or more challenge-response pairs. Note that this definition merely states what qualifies for an ORAS and does not reflect on the security of the ORAS. Indeed, an ORAS might only be secure for a few observations, before the secret can be extracted. For an ORAS to be secure, the probability of finding the secret should be small (or negligible) for a large number of challenge-response pairs. Since each challenge-response pair leaks some information about the secret, the goal of the designer is to use the ORAS for as many challenge-response rounds as possible before the adversary can extract the secret with non-negligible probability. Note that password-based authentication is not observation resilient under this definition, as the secret is recovered after one observation.

\begin{figure}[t!]
	\centering
    \includegraphics[width=0.60\columnwidth]{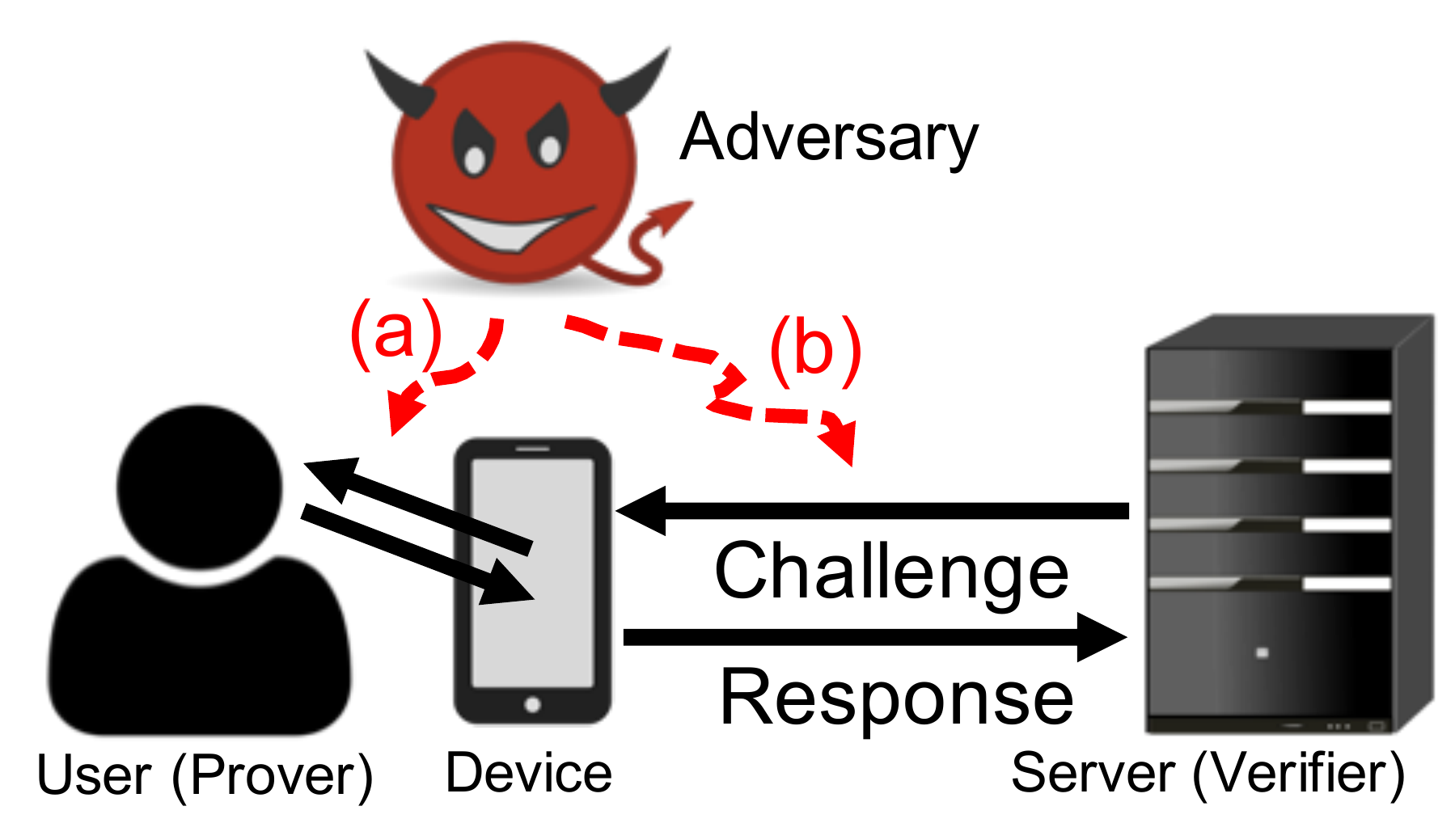}
    \vspace{-2mm}
  	\caption{The threat model under consideration. Adversary can also observe the interaction between the user and the device.}
 	\label{fig:process_diagram}
 	\vspace{2mm}
\end{figure}

\subsection{\texorpdfstring{$k$-out-of-$n$}{k-out-of-n} ORAS}
\label{sec:oras_scope}
In one class of ORAS the secret $S$ is a random set of $k$ items from a set of $n$ (publicly known) items\footnote{Examples of items are images~\cite{sasamoto2008undercover, dhamija2000deja, brostoff2000passfaces} or emoticons~\cite{chauhan2017behaviocog, asghar2013does, wiedenbeck2006design}.}. The elements of $S$ are called the secret items, and the remaining $n - k$ items will be referred to as decoy items. We shall call these $k$-out-of-$n$ ORAS. Different designs of these ORAS exist. In the following we focus on particular design elements which are promising in terms of both resistance to known attacks and usability (especially if employed in conjunction with other authentication factors such as behavioral biometrics). These design elements are:
\begin{itemize}
    \item \emph{Windowed Challenges:} A challenge is constructed by randomly selecting $l$ out of $n$ items. This can be visualized as a fixed window capable of enclosing $l$ items. The set of $n$ items is randomly shuffled each time, and the $l$ items within the window are the challenge items (hence the name). If the \emph{window size} $l$ is small, the user can recognize its secret items present in the challenge in a short amount of time. This is also desirable for deployment as a small set can be easily displayed on the small screens of smartphones~\cite{chauhan2017behaviocog}. However, if not designed carefully, these windowed challenges may compromise security. For instance, the Undercover \cite{sasamoto2008undercover} scheme requires \emph{at least} one secret item to be present in all challenges. This results in an inherent bias, with the secret items appearing more frequently than the decoy items. This bias was exploited by Yan et al.~\cite{yan2012limitations} in a frequency analysis attack to extract the entire set of secret items after only a small number of observations. Subsequently, Asghar et al.~\cite{asghar2013does} showed that if the windowed challenge of length $l$ is sampled uniformly at random from the set of all possible $\binom{n}{l}$ challenges, then the above mentioned frequency-based attack can be mitigated. 
    \item \emph{Random Weights:} Each of the $l$ items in the challenge is associated with a random integer, called its \emph{weight}, from the set $\mathbb{Z}_d$, for a fixed integer $d \ge 2$. Note that for each challenge the weights are randomly sampled anew.
    \item \emph{Modulus Operation:} The function $f$, to be mentally computed by the user, involves (at the minimum) summing the weights of the secret items present in the challenge and a modulo $d$ operation on the sum. Notice that by construction, a challenge might not even contain any of the $k$ secret items. We shall refer to it as the \emph{empty event} or \emph{empty case}, borrowing the term from~\cite{chauhan2017behaviocog}. How the function $f$ is computed in an empty case depends on the scheme, as we shall discuss shortly.
\end{itemize}

\begin{example} 
\label{ex:toy}
We illustrate a $k$-out-of-$n$ ORAS that satisfies the above design requirements. 
A windowed challenge can be represented by the $n$-element vector $\mathbf{c}$ whose $i$th element is the weight of the $i$th item, if present in the challenge, and 0 otherwise. With the same ordering, the secret can be represented as the binary vector $\mathbf{s}$ of Hamming weight $k$. One possible cognitive function $f$ is the dot product modulo $d$, i.e, the response $r$ is calculated as $\ip{\mathbf{c}}{\mathbf{s}} \bmod d$. 
Consider the example in Figure \ref{fig:sample_challenge}, which has a pool of $n=8$ items, with $k=3$ secret items. A random challenge of $l=4$ items has been sampled, with accompanying random values from $\mathbb{Z}_d = \mathbb{Z}_4$. The user computes $r = \ip{\mathbf{c}}{\mathbf{s}} \bmod d = \ip{\begin{pmatrix} 0 & 0 & 1 & 0 & 0 & 0 & 3 & 2\end{pmatrix}}{\begin{pmatrix} 0 & 1 & 1 & 0 & 0 & 0 & 1 & 0\end{pmatrix}} \\ \bmod 4 = 0$. The scheme is observation resilient as there are multiple candidates for the secret even after observing the challenge and this response. For example, the response 0 could have simply come from the weight of item D.
\end{example}
\begin{figure}[th!]
	\centering \includegraphics[width=0.50\columnwidth]{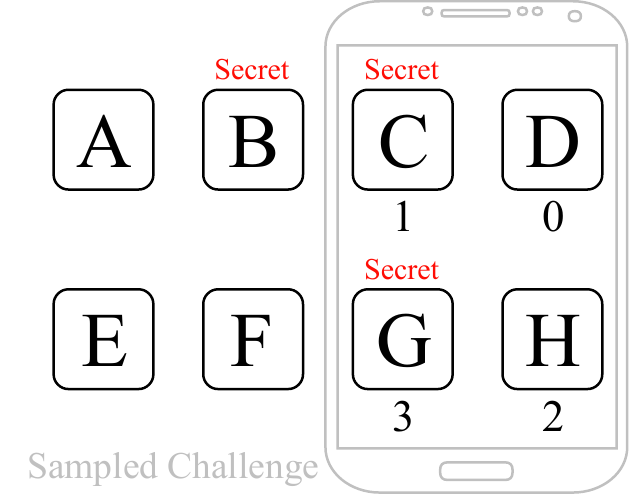}
  	\caption{Example ORAS Scheme, $(n,k,l,d) = (8,3,4,4)$}
 	\label{fig:sample_challenge}
\end{figure}

\subsection{\texorpdfstring{$k$-out-of-$n$}{k-out-of-n} ORAS Chosen for Analysis}
\label{sec:existing_protocols}
We introduce three previously proposed ORAS that fall in the category of $k$-out-of-$n$ ORAS described in Section~\ref{sec:oras_scope}. Note that with $O(n)$ challenge-response pairs in Example~\ref{ex:toy}, the attacker can construct the secret using Gaussian elimination. The cognitive function in the three protocols is designed to increase the challenge-response pairs required to recover the secret through Gaussian elimination.

\paragraph{BehavioCog (BC)}
The BehavioCog~\cite{chauhan2017behaviocog} scheme is the same as in Example~\ref{ex:toy}, except that it requires the user to submit a random response $r \in \mathbb{Z}_d$ in the case of an empty event (i.e., when none of the secret items are in the challenge). With this modification, Gaussian elimination requires $O(dn)$ challenge-response pairs~\cite{chauhan2017behaviocog}. One set of proposed parameters for the scheme is $(n, k, l, d) = (180, 14, 30, 5)$~\cite{chauhan2017behaviocog}, which we shall use in our analysis. We remark that BehavioCog was proposed with a behavioral biometric component to minimize authentication time. We disregard the biometric component and focus on the cognitive scheme.

\paragraph{FoxTail (FT)}
We chose the FoxTail scheme proposed by Asghar et al. ~\cite{asghar2013does} as a fix to secure the original FoxTail scheme~\cite{li-shum} against a frequency attack~\cite{yan2012limitations}. 
The parameters we use for the scheme are $(n, k, l, d) = (180, 14, 30, 4)$ 
to allow for comparison between schemes. The cognitive function involves an additional step after the modulo $d = 4$ operation: the user is required to respond with 0, if the result is 0 or 1, and respond with 1, if the result is 2 or 3. In the case of an empty event, the user simply returns the response 0. The resulting non-linear map means that Gaussian elimination through linearization requires $\binom{n}{2} + n = O(n^2)$ challenge-response pairs~\cite{asghar2015linearization}.

\paragraph{Hopper \& Blum (HB)}
The HB protocol~\cite{hopper2001secure} is one of the earliest ORAS proposed. The original proposal displays all $n$ items to the user, with accompanying random binary weights. We modify the scheme to utilize windowed challenges, choosing parameters $(n, k, l, d) = (180, 14, 30, 2)$ (similar to BehavioCog). The protocol requires the user to intentionally flip the response bit (note that $d = 2$) with a fixed probability $\eta < 0.5$. The windowless HB was subject to timing attacks with the noise parameter $\eta = 0.2$ \cite{vcagalj2015timing}; this value is maintained in this paper.\footnote{We note that the timing attack from \cite{vcagalj2015timing} is not applicable to the windowed HB protocol.} The HB protocol is based on the NP-Hard problem of learning parity in the presence of noise (LPN).

\paragraph*{Rounds and Sessions} Each authentication session consists of multiple challenge-response rounds. The number of rounds per session can be selected based on the success probability of randomly guessing the response (without knowledge of the secret). One common benchmark is 6-digit PIN, with a probability of randomly guessing the correct pin being $P_{\text{RG}} = 10^{-6}$~\cite{chauhan2017behaviocog}. In BC (with the above parameters), the attacker can successfully guess the response to a challenge with probability $0.256$. Therefore, to achieve the security level of $10^{-6}$, 10 rounds are required in a session. Similarly for FT, the answer could be guessed with probability $0.5$, thus requiring 20 rounds per session for the same security level. In HB protocol, the user is accepted if the fraction of wrong answers are at most $\eta$~\cite{hopper2001secure}. For $\eta = 0.2$, this gives $51$ rounds for a security level of $10^{-6}$. Since this number of rounds is impractical, we lower the security level for HB to $10^{-4}$, which gives us $34$ rounds per session.
In the remainder of this paper we will mostly use the number of rounds instead of sessions to discuss attacks, as the number of rounds within a session is ultimately at the discretion of the scheme designer.

\subsection{Other ORAS}
While we use $k$-out-of-$n$ ORAS as the basis for our analysis, the results are applicable to other ORAS. More specifically, the results are applicable to any ORAS that uses a modulus operation. In Section~\ref{sec:other_oras}, we shall give examples of these ORAS and our behavioral side channel attack on them.

\section{The Modulus Event and Associated Biases}
\label{sec:mod_bias}

Given a challenge, we say that a modulus event occurs if the submitted response involves a modulus operation. For $k$-out-of-$n$ ORAS, this happens if the sum of the secret items in a challenge is greater than $d$ (otherwise the user does not need to reduce the sum modulo $d$). In this section, using a generic $k$-out-of-$n$ ORAS, we show that
\begin{enumerate}
    \item Depending on the parameters and the cognitive function, there is an imbalance in the likelihood of a modulus or a no-modulus event given different response values, e.g., a response of 0 is more likely to indicate a modulus event.
    \item In a no-modulus event, the secret items have lower weights than the decoy items. Likewise in a modulus-event, the reverse is true.
\end{enumerate}
The first of these observations will be used as one of the features to determine a modulus/non-modulus event in our classifiers in Section~\ref{sec:classifier}. The second observation is the basis of our algorithm to retrieve the secret in Section \ref{sec:algorithms}. While the response itself indicates if the modulus event has occurred or not, user behavior while computing the function $f$ leaks further information about the event. This can be exploited by the adversary to increase confidence in predicting the modulus/no-modulus event to retrieve the secret.

Since the function $f$ is mentally computed by the user, the adversary cannot know if the modulus event has occurred by simply looking at challenge-response pairs. However, user behavior while computing $f$ leaks information about these events. 

In what follows, we mathematically demonstrate that given a generic $k$-out-of-$n$ ORAS, both above mentioned biases pertaining to the modulus event are linked to the (expected) number of secret items present in a challenge. The lower the number, the bigger the bias. Since this number is a function of the parameters $(n, k, l)$, scheme designers need to choose appropriate values of these parameter to ensure that the expected number of secret items is large to minimize the biases.

\subsection{Guessing a Modulus Event through Responses}
Let us demonstrate this bias with the help of an example. Let $G$ be the random variable representing the number of secret items present in a challenge. Thus, $G$ takes on values in the set $\{0, 1, \ldots, k\}$, where $k$ is the total number of secrets items. Let $g$ denote an instance of $G$ in a specific round of authentication.
\begin{example}
In Table \ref{tab:3_sec_weights}, the number of secret item's present in the challenge is $g=3$, and the responses are generated through the cognitive function in  Example~\ref{ex:toy} with the modulus $d=2$. Every combination of weights is equally probable (due to random sampling of weights). It is evident from the table that when the response is 0, it is more likely to be the result of having performed the modulus operation; whereas when the response is 1, it is more likely to be due to the absence of the modulus operation.
\end{example}

\begin{table}[t]
\centering
\caption{Modulo Bias in responses, 3 Secret Items, Binary weights.}
\label{tab:3_sec_weights}
\vspace{-2mm}
\resizebox{0.6\columnwidth}{!}{%
\begin{tabular}{| l | c | c |c |c |c |c |c |c |}
\hline

Secret Weight 1 & 0 & 0 & 0 & 0 & 1 & 1 & 1 & 1 \\
Secret Weight 2 & 0 & 0 & 1 & 1 & 0 & 0 & 1 & 1 \\
Secret Weight 3 & 0 & 1 & 0 & 1 & 0 & 1 & 0 & 1 \\ \hline
Modulus Event & \cellcolor{yellow}No & \cellcolor{yellow}No & \cellcolor{yellow}No & \cellcolor{green}Yes & \cellcolor{yellow}No & \cellcolor{green}Yes & \cellcolor{green}Yes & \cellcolor{green}Yes \\
Resulting Sum & \cellcolor{yellow}0 & \cellcolor{yellow}1 & \cellcolor{yellow}1 & \cellcolor{green}2 & \cellcolor{yellow}1 & \cellcolor{green}2 & \cellcolor{green}2 & \cellcolor{green}3 \\
User Response & \cellcolor{yellow}0 & \cellcolor{yellow}1 & \cellcolor{yellow}1 & \cellcolor{green}0 & \cellcolor{yellow}1 & \cellcolor{green}0 & \cellcolor{green}0 & \cellcolor{green}1 \\
\hline
\end{tabular}
\vspace{2mm}
}
\end{table}

We now generalize this to a generic ORAS.
Let the random variable $X$ denote the weight of an item in a challenge. Since each weight is sampled from a random uniform distribution over $d$, 
\[
\Pr(X = x) = \frac{1}{d}, \text{ for all } x \in \{0,1, \ldots , d-1\}.
\]
Let $Y$ be the random variable denoting the sum of the weights of the $g$ secret items. Then $Y$ takes on values from the set
\[
\{0,1,2, \ldots, (d-1)g\}.
\]
We would like to determine $\Pr(Y = y \mid g)$, from which we can determine the probability of a modulus event by evaluating $\Pr(Y \ge d \mid g)$. To compute $\Pr(Y=y \mid g)$, we need to find the number of different ways $g$ items with weights in $\mathbb{Z}_d$ can be summed to produce $y$. This is determined by the coefficient of $z^y$ in the expansion of the polynomial $(z^0 + z^1 + \cdots + z^{d-1})^g$ \cite[\S 1, pp. 23-24]{uspensky1937introduction}. 
Alternatively the probability can be evaluated without full expansion of the generating function via the following equation \cite[\S 1, p. 24]{uspensky1937introduction}:
\begin{equation}
\Pr(Y=y \mid g) = \frac{1}{d^{g}} \sum_{s=0}^{\lfloor \sfrac{y}{d} \rfloor} (-1)^s \binom{g}{s} \binom{y - sd + g -1}{g-1}.
\label{eqn:prob_sum}
\end{equation}
Note that a modulus operation is \emph{not} required when $Y=y < d$,
whose probability is given by
\begin{equation}
\label{eqn:no-mod-event}
\Pr(Y<d \mid g) = \sum_{y=0}^{d-1} \Pr(Y=y \mid g).
\end{equation}
Similarly the probability that a modulus operation is required is
\begin{equation}
\label{eqn:mod-event}
\Pr(Y\ge d \mid g) = \sum_{y=d}^{(d-1)g} \Pr(Y=y \mid g).
\end{equation}
Equation \ref{eqn:prob_sum} is dependent on $g$; in turn the probability of $g$ items appearing in a challenge is dependent on $n, k, \text{and } l$ of the authentication scheme:
\begin{equation}
\Pr(G = g) = \frac{\binom{n-k}{l-g}\binom{k}{g}}{\binom{n}{l}}.
\label{eqn:prob_g}
\end{equation}

In Figure~\ref{fig:mod_given_response_toy}, we plot the probabilities of the modulus and no-modulus event given different values of $g$ calculated through Eqs.~\ref{eqn:no-mod-event} and \ref{eqn:mod-event}. We see that if $g$ is small, the no-modulus event is highly probable and a given response value might be biased towards a modulus or a no-modulus event. While this bias does not directly reveal information about the user's secret items, it can be used as an indicator of a particular challenge involving a modulus or a no-modulus event. And, as we shall show in the next section (Section~\ref{sec:weight_bias}), the knowledge of a modulus/no-modulus event, in turn, leaks information about the user's secret. Thus, from a security point of view, it is desirable to minimize this bias. This can be done by increasing the value of $g$, which makes the no-modulus event increasingly unlikely to happen (irrespective of the response). Asymptotically, we have the following result:

\begin{theorem}
As $g \to \infty$, the probability of the modulus event approaches one, i.e., $\Pr(Y\ge d) \to 1$.
\label{theorem:prob_mod_1}
\end{theorem}
\begin{proof}
Please see Appendix~\ref{proof:prob_mod_1}.
\end{proof}

In practice the likelihood of the no-modulus event vanishes much rapidly with an increasing $g$. For instance, for the case of $d = 2$, the probability is $0.1875$ with $g = 5$, $0.0107$ with $g = 10$, and $0.0005$ with $g = 15$. We remark that the number of items present in a challenge $g$ is a random variable dependent on the scheme parameters. Thus, to ensure the responses do not exhibit a bias towards the modulus or no-modulus event, the expected value of $g$ needs to be high in an ORAS, which can be done by a combination of increasing $k$ or $l$, and/or decreasing $n$ (c.f. Equation 2). We will return to this in Section~\ref{sec:discussion}.

\begin{figure}[t]
	\centering
    \includegraphics[width=0.6\columnwidth]{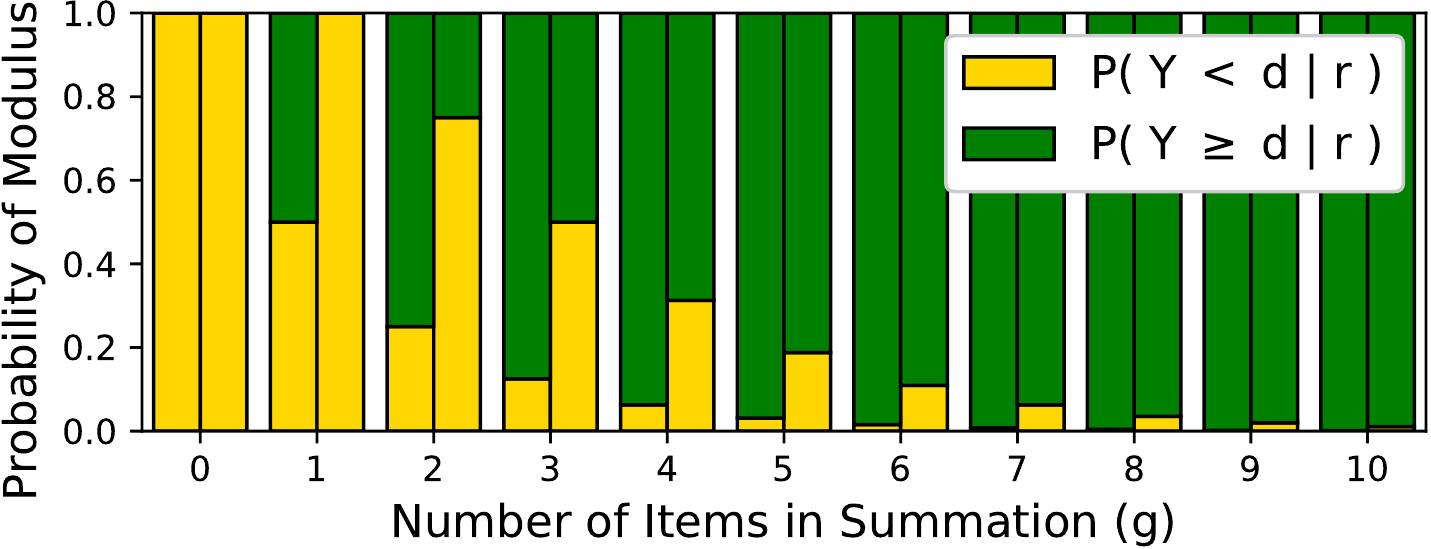}
    \vspace{-2mm}
  	\caption{The probability of the modulus and no-modulus events given a user response against the number of secret items present in a challenge $g$ when $d = 2$. The probability of modulus event increases with increasing $g$. The left and right hand columns respectively represent a user response of 0 and 1.}
 	\label{fig:mod_given_response_toy}
 	\vspace{2mm}
\end{figure}

\subsection{Weight Bias in a Modulus Event}
\label{sec:weight_bias}
We now show that given a modulus or a no-modulus event the expected weight of the secret items is biased away from the expected weight of the decoy items.
Thus, the knowledge of a modulus/no-modulus event leaks information about the secret items. We will use this observation in our algorithm to retrieve the secret in Section~\ref{sec:algorithms}. For now we demonstrate this bias analytically using a generic $k$-out-of-$n$ ORAS.

Recall that $X$ denotes the weight of an item in a challenge. Clearly, the expected weight of any item within a challenge is $E[X] = (d-1)/2$. Denote by $X_s$, the random weight of a secret item. By construction, we have $E[X_s] = E[X] = (d-1)/2$. However, given the knowledge of a no-modulus event, the conditional expectation might not be the same. To see this, first note that
\begin{align}
E[X] = E[X_s] &= E[X_s \mid Y < d]\Pr(Y < d) \nonumber\\
            &+ E[X_s \mid Y \ge d]\Pr(Y \ge d), \label{eq:exptotsec}
\end{align}
where the conditional expectations are conditioned by no-modulus and modulus events, respectively. Since $Y$ denotes the additive weight of $g$ secret items, we have
$$E[X_s] = \frac{1}{g} E[Y].$$
We first show that the expected weight of secret items is less than or equal to the expected weight of decoy items in a no-modulus event, i.e., $E[X_s \mid Y < d] \le E[X]$. The following lemma is used in the proof.
\begin{lemma}
Let $g \ge 2$ be an integer, and let $p$ be a strictly positive function, i.e., $p(i) > 0$, for all $i$ in the domain of $p$. Then, for all $d \ge 1$ 
\[
\frac{1}{g} \sum_{i=0}^{d} i p(i) < \frac{d}{2} \sum_{i=0}^{d}  p(i).
\]
\label{lem:d-ineq}
\end{lemma}%
\begin{proof}
See Appendix~\ref{sec:lem1_proof}.
\end{proof}
\begin{theorem} 
\label{theo:weight_bias}
Let $g \ge 1$. Then
\begin{enumerate}
    \item $E[X_s \mid Y < d] = E[X]$, if $g = 1$.
    \item $E[X_s \mid Y < d] < E[X]$, if $g \ge 2$.
\end{enumerate}
\end{theorem}
\begin{proof}

For part (1), when $g = 1$, only one secret item is present, thus it's weight will be uniformly sampled from $0$ to $d-1$, and therefore in this case, $E[X_s | Y < d] = E[X]$. 

For part (2), we have
\begin{align*}
    E[X_s \mid Y<d] &= \frac{1}{g}E(Y \mid Y<d) \\
                    &= \frac{1}{g} \sum_{y} y \Pr(Y = y \mid Y < d) \\
                    &= \frac{1}{g} \sum_{y} y \frac{\Pr(Y = y, Y < d)}{\Pr(Y < d)} \\
                    &= \frac{1}{g} \sum_{y=0}^{d - 1} y \frac{\Pr(Y = y)}{\Pr(Y < d)}.
\end{align*}
Invoking Lemma~\ref{lem:d-ineq}:
\begin{align*}
    E[X_s \mid Y<d] &< \frac{1}{\Pr(Y < d)}\frac{d-1}{2} \sum_{y=0}^{d - 1} {\Pr(Y = y)} \\
                    &= \frac{d-1}{2} \frac{\sum_{y=0}^{d-1}  P(Y=y)}{\sum_{y=0}^{d-1}  P(Y=y)}
                    = E[X].\tag*{\qedhere}
\end{align*}\qedhere
\end{proof}
From this, it follows that expected weight of secret items is strictly greater than the expected weight of decoy items in a modulus event. That is:
\begin{corollary}
\label{cor:ex-mod}
$E[X_s \mid Y \ge d] > E[X]$.
\end{corollary}
\begin{proof}
When $Y \ge d$, i.e., a modulus event, $g$ is necessarily $\ge 2$. From Eq.~\ref{eq:exptotsec} and Theorem~\ref{theo:weight_bias} part (2), we have

\begin{align*}
   & E[X] < E[X] \Pr(Y < d) + E[X_s \mid Y \ge d]\Pr(Y \ge d) \\
 \Rightarrow  &E[X] (1 - \Pr(Y < d)) < E[X_s \mid Y \ge d]\Pr(Y \ge d) \\
 \Rightarrow  &E[X] < E[X_s \mid Y \ge d] \qedhere
\end{align*}
\end{proof}

Thus, given a modulus event, weights of the secret items tend to be higher than decoy items. The reverse is true in a no-modulus event. We note that when $g = 0$, i.e., the empty event, the expected weight of the secret items is undefined (since they do not exist in the challenge).

Recall from Theorem~\ref{theorem:prob_mod_1} that as $g$ increases, the probability of the no-modulus event approaches zero. Therefore, increasing $g$ should also minimize the weight bias. This is demonstrated by the following corollary which shows that as $g$ increases, the expected weight of a secret item approaches the global expectation $E[X]$.

\begin{corollary}
\label{cor:no-weightbias}
As $g \to \infty$, $E[X_s \mid Y \ge d] \to E[X]$.
\end{corollary}
\begin{proof}
From Theorem~\ref{theorem:prob_mod_1}, we have $\Pr(Y \ge d) \to 1$ as $g \to \infty$. Consequently, $\Pr(Y < d) \to 0$. The result then follows from Eq.~\ref{eq:exptotsec}.
\end{proof}

Thus, scheme designers can minimize any weight bias by increasing the parameters $(k, l)$, or decrease ($n$) to increase the expected value of $g$. Once again the above result is asymptotic, and in practice the bias in expected weights vanishes quickly by a moderate increase in $g$ since the probability of the no-modulus event decreases quickly. Unfortunately, the three ORASes under focus, do not have a sufficiently large expected value of $g$, and are susceptible to revealing information about the secret items through this bias.

\subsection{Biases in Specific ORAS}

\begin{figure}
  \begin{subfigure}[t]{\columnwidth}
    \centering
    \includegraphics[width=0.6\columnwidth]{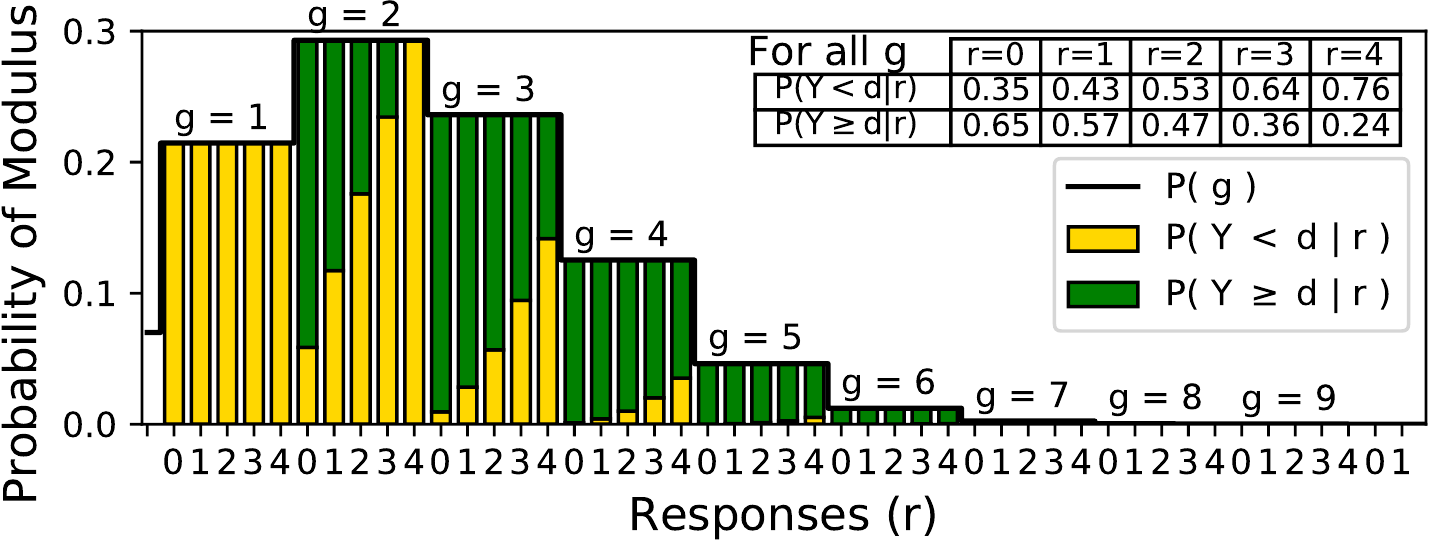}
    \vspace{-2mm}
    \caption{BehavioCog Protocol $(180, 14, 30, 5)$}
    \label{fig:BC_mod_bias}
  \end{subfigure}
  \vspace{2mm}
  \begin{subfigure}[b]{\columnwidth}
    \centering
    \includegraphics[width=0.6\columnwidth]{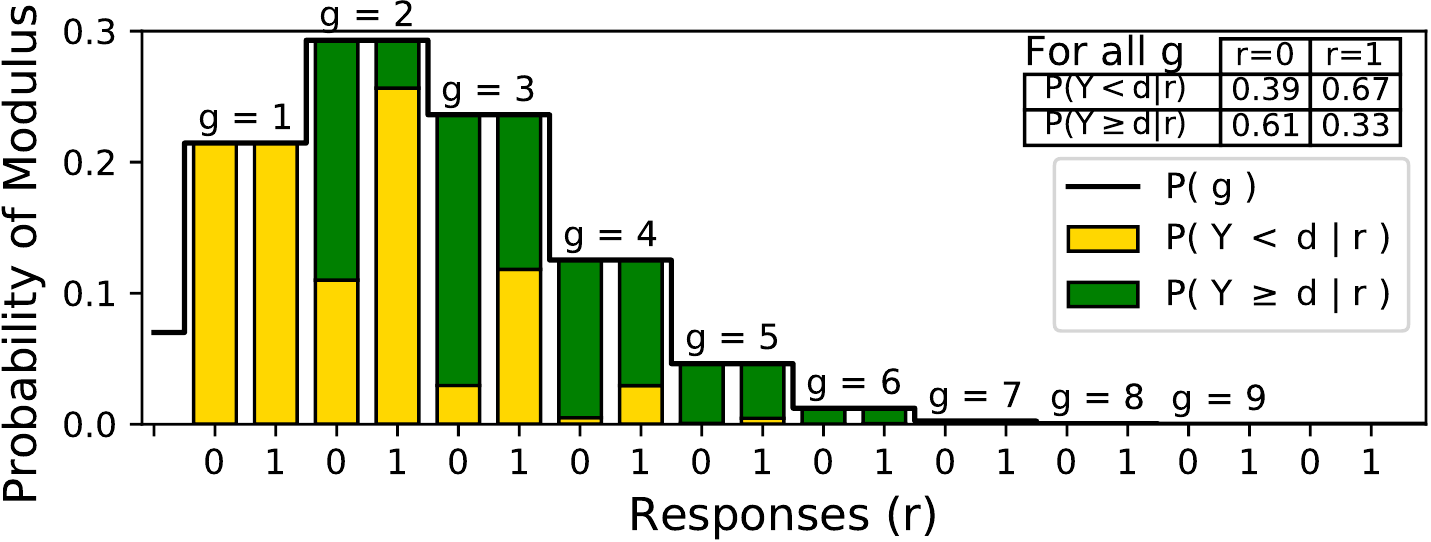}
    \vspace{-2mm}
    \caption{Foxtail Protocol $(180, 14, 30, 4)$}
    \label{fig:FT_mod_bias}
  \end{subfigure}
  \vspace{2mm}
  \begin{subfigure}[b]{\columnwidth}
    \centering
    \includegraphics[width=0.6\columnwidth]{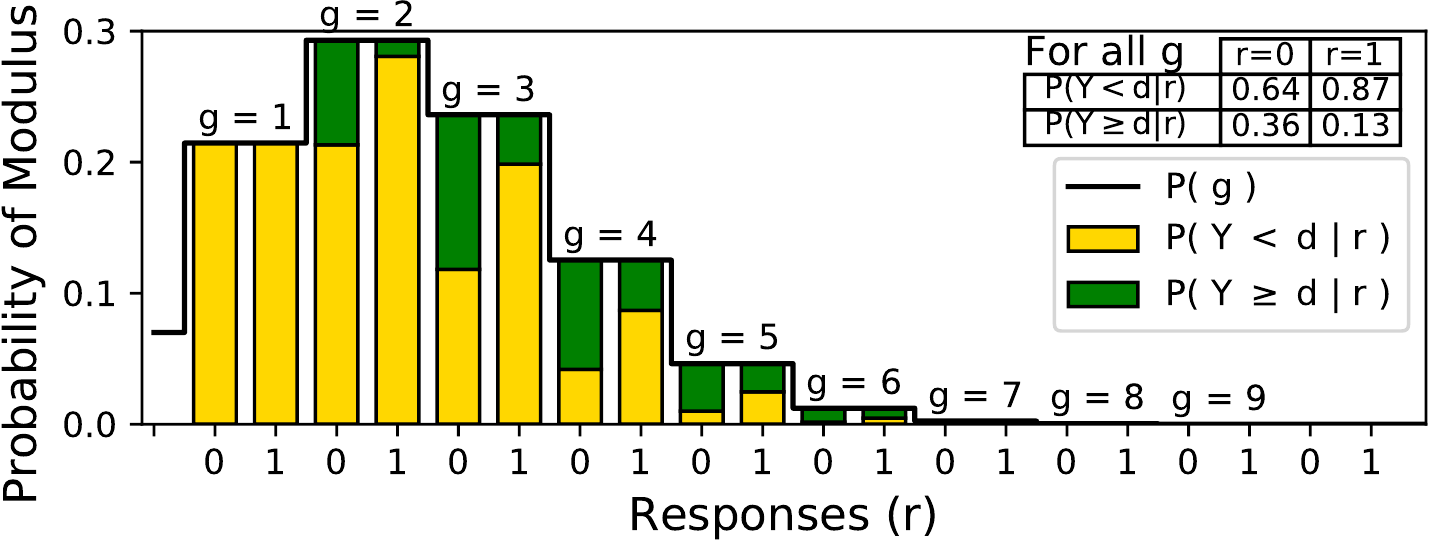}
    \vspace{-2mm}
    \caption{HopperBlum Protocol $(180, 14, 30, 2)$}
    \label{fig:HB_mod_bias}
  \end{subfigure}
  \vspace{-7mm}
  \caption{Probabilities of modulus and no-modulus events given different response values in $k$-out-of-$n$ ORAS. The probabilities are given for different values of $g$, showing the normalized probability for the event of the modulus operation. The overall probabilities irrespective of $g$ are given in the table.}
  \label{fig:mod_bias}
  \vspace{2mm}
\end{figure}

We now highlight these biases in the three instances of ORAS: BehavioCog, FoxTail, and HopperBlum. The cognitive functions in each of these schemes are slightly more involved than a simple mod operation on the sum (e.g., response flipping in HB). Therefore, the analytical results on the generic ORAS may not completely reflect the biases in these schemes.

\subsubsection{Guessing a Modulus Event through Responses}
\label{sec:mod_bias_response}

First consider the BehavioCog scheme with parameters $(n, k, l, d) = (180, 14, 30, 5)$. Figure~\ref{fig:BC_mod_bias} shows the probability of a modulus and a no-modulus event given different response values broken down across different values of $g \ge 1$. The inset table in the figure shows the two probabilities irrespective of the value of $g$. Clearly, a higher response value indicates that it is more likely a no-modulus event. In contrast, lower response values are more likely to be a result of a modulus event. Not surprisingly, the probability of the no-modulus event decreases with increasing $g$. Recall that in BehavioCog, the user enters a random response in case of the empty event. However, the bias shown in the figure is for $g \ge 1$. The inset table on the other hand shows includes the empty event as well, and hence shows probabilities for all $g$.

The corresponding biases in the Foxtail and HB protocols are shown in Figures~\ref{fig:FT_mod_bias} and \ref{fig:HB_mod_bias}, respectively. Recall that both Foxtail and HB have the response space $\{0, 1\}$. In Foxtail the response is ``rounded'' after a mod 4 operation, and in HB it is flipped with a fixed probability $\eta$ (which we fix to $0.2$). In Foxtail, we see that the response 0 is more likely to be from a non-modulus event, whereas the reverse is true of response 1. On the other hand, both responses in HB are more likely due to a no-modulus event, with response 1 being considerably more biased towards the no-modulus event.
Thus, the bias of a particular response towards a modulus event depends on the scheme parameters as well as the cognitive function.

\subsubsection{Weight bias in a Modulus Event}
Table~\ref{tab:e_weight} shows the expected weight of secret item(s) given a modulus and a no-modulus event for all three schemes. The expected weights are also shown against the number of secret items present in the challenge. The case $g = 1$ obviously does not involve any modulus operation, and so the expected weight in this case is equal to overall expectation. This is shown in the table with the row labelled $E[X_s]$. In all three protocols, the expected weights of the secret items given a no-modulus event is lower than the overall expectation, and decreases further as the number of secret items $g$ increases. However, this also means that the no-modulus event becomes almost unlikely to occur. The expected weight of the secret items in a modulus event is higher for smaller values of $g$ and approaches the overall expectation as we increase $g$. Noting in Figure~\ref{fig:mod_bias} that higher values of $g$, say $g \ge 4$, the bias is less profound and is less likely to occur in a challenge (with the given parameters).

\begin{table}[t]
\centering
\caption{Expected weights $E(X_s)$ of secret items.}
\label{tab:e_weight}
\vspace{-2mm}
\resizebox{0.6\columnwidth}{!}{%
\begin{tabular}{|@{ }l@{ }r@{ }|@{ }c@{ }|@{ } l@{\quad}l@{\quad}l@{\quad}l@{\quad}l@{\quad} l@{\quad}l@{\quad}l@{\quad}l@{ }|}
\hline
\multicolumn{2}{|@{ }c@{ }|@{ }}{$E(X_s | Y,g)$} & All $g$ & $g$=1    & $g$=2    & $g$=3    & $g$=4    & $g$=5    & $g$=6    & $g$=7    & $g$=8    & $g$=9    \\ \hline
BC & $Y<d$ & 1.20    & 2.00 & 1.33 & 1.00 & 0.80 & 0.67 & 0.57 & 0.50 & 0.44 & 0.40 \\
$d$=5                            & $Y\geq d$   & 2.57    & -    & 3.00 & 2.39 & 2.15 & 2.06 & 2.02 & 2.01 & 2.00 & 2.00  \\ 
& $E[X_s]$ & 2.00 & 2.00 & 2.00 & 2.00 & 2.00 & 2.00 & 2.00 & 2.00 & 2.00  & 2.00 \\
\hline 
FT    & $Y<d$ & 0.90    & 1.50 & 1.00 & 0.75 & 0.60 & 0.50 & 0.43 & 0.38 & 0.33 & 0.30 \\
$d$=4                            & $Y\geq d$   & 1.98    & -    & 2.33 & 1.84 & 1.64 & 1.56 & 1.52 & 1.51 & 1.50 & 1.50  \\ 
& $E[X_s]$ & 1.50 & 1.50 & 1.50 & 1.50 & 1.50 & 1.50 & 1.50 & 1.50 & 1.50  & 1.50 \\
\hline 
HB & $Y<d$ & 0.30    & 0.50 & 0.33 & 0.25 & 0.20 & 0.17 & 0.14 & 0.13 & 0.11 & 0.10 \\
$d$=2                            & $Y\geq d$   & 0.82    & -    & 1.00 & 0.75 & 0.64 & 0.58 & 0.54 & 0.53 & 0.51 & 0.51 \\ 
& $E[X_s]$ & 0.50 & 0.50 & 0.50 & 0.50 & 0.50 & 0.50 & 0.50 & 0.50 & 0.50  & 0.50 \\
\hline
\end{tabular}
}
\vspace{2mm}
\end{table}


\section{Attack Algorithm and the Faulty Oracle}
\label{sec:algorithms}
In the previous section, we discussed how knowledge of the modulus or no-modulus event leaks information about the secret. In this section, we will construct an attack algorithm that retrieves the secret given access to an oracle that indicates a modulus or no-modulus event. A ``perfect'' oracle, however, is unrealistic in practice where we expect some error in our knowledge of the event. We therefore assume a \emph{faulty} oracle which might erroneously indicate a modulus event. We will analyze the performance of the attack algorithm against varying accuracies of the faulty oracle. 

More precisely we consider a faulty oracle, denoted $\mathcal{O}_{\text{mod}}$, which when given a challenge (and any auxiliary information) as input, returns $-1$ if it guesses that the user has \emph{not} performed the modulus operation (i.e., the no-modulus event is the positive class), and $+1$ otherwise. The oracle can make two types of errors; \textit{type 1 error} is when the oracle outputs $-1$ when it is a modulus event, and a \textit{type 2 error} is when the oracle outputs $+1$ when in fact it is actually a no-modulus event. The true positive rate (TPR) is the probability of correctly guessing the modulus event, and therefore, $1 - \text{TPR}$ is the probability of \textit{type 1 error}. Similarly, the true negative rate (TNR) is the probability of correctly guessing the no-modulus event. Thus, $1 - \text{TNR}$ is the probability of \textit{type 2 error}. The oracle is parameterized by these two probabilities and we denote this by $\mathcal{O}_{\text{mod}}^{\text{TPR}, \text{TNR}}$.

Algorithm \ref{algo:mod_point_update} describes our algorithm to retrieve the secret with access to this faulty oracle, which we call the Modulus Event Points Update algorithm. The algorithm maintains a list of points $(p_1, \ldots, p_n)$ where $p_i$ denotes the points for item $i$. Initially, all items have a score of 0. Upon receiving a challenge, the algorithm consults the faulty oracle. 
If the faulty oracle detects a no-modulus event, then it penalizes all items whose weights are greater than the response $r$ (since such items would require a modulus event to produce the response $r$). Here items with higher weights are given higher penalties (as the expected weight of secret items is lower in the no-modulus event). When the oracle detects a modulus event, items with lower weights are given higher penalties as secret items are expected to have higher weights in a modulus event. This is reflected in the construction of the penalty vector $(v_0, \ldots, v_{d-1})$, where we have $v_0 \le \cdots \le v_{d-1}$. Note that in this case items are penalized irrespective of the response. This is because an item with any weight could have produced the response (since it is a composite of multiple weights reduced modulo $d$.

Since the oracle is faulty, secret items may also get penalized. However, the decoy items are penalized more with an increasing number of challenges, eventually leading to higher scores for the secret items. To show this, we consider a special case of the algorithm and show that the expected score of a secret item is higher than a decoy item with large enough $m$, i.e., number of challenge-response pairs. Specifically, we consider the penalty vector $(v_0, \ldots, v_{d-1})$ to be all zeroes, i.e., no points update in case of the modulus event. Our simulations show that not updating the points at all when a modulus event is detected does indeed take the least number of samples to retrieve the secret.
\begin{theorem}
\label{the:no-mod-point-update}
Let the penalty vector $(v_0, \ldots, v_{d-1})$ be all zeroes. Furthermore, let the other penalty vector, i.e., $(u_0, \ldots, u_{d-1})$  be not identically zero. If $\text{TNR} > 1 - \text{TPR}$, then for sufficiently large $m$, the expected score of a secret item is more than the score of a decoy item.
\end{theorem}

\begin{proof}
See Appendix \ref{proof:no-mod-point-update}.
\end{proof}

\begin{algorithm}[t]
\SetAlgoLined
\DontPrintSemicolon
\KwIn{Scheme parameters $(n, k, l, d)$, number of challenges $m$, penalty vectors $(u_0, \ldots, u_{d-1})$ where $0 = u_0 \ge \cdots \ge u_{d-1}$, and $(v_0, \ldots, v_{d-1})$, where $v_0 \le \cdots \le v_{d-1}$.}
\KwOut{A list of points $( p_1, p_2, \ldots, p_n )$, with top $k$ scores indicating secret items.}
Initialize $( p_1, p_2, \ldots, p_n)$ to all zeroes.\;
\For{$j = 1$ to $m$}{
    Observe challenge $c$ containing items $i$ and weights $\textnormal{wt}(i)$, auxiliary information `aux,' and response $r$.\;
    $b \gets \mathcal{O}_{\text{mod}}^{\text{TPR}, \text{TNR}}(c, \text{aux})$.\;
    \If{$b = -1$ (no-modulus event)}{
        \For{all items $i$ such that $\textnormal{wt}(i) > r$}{
            penalize $p_i \gets p_i + u_{\text{wt}(i)}$.\;
        }
    }
    \Else{
            penalize $p_i \gets p_i + v_{\text{wt}(i)}$.\;
    }
}
\Return{$( p_1, p_2, \ldots, p_n)$}.\;
\caption{{\sc Modulus Event Points Update}}
\label{algo:mod_point_update}
\end{algorithm}

\begin{table*}[t]
\raggedright
\caption{Experimentally derived rounds required to reveal full user secret given varying TPR and TNR of side-channel classifier, for the Modulus applied on BehavioCog (900 Round Benchmark \cite{chauhan2017behaviocog}).}
\label{tab:mod_algo_results_bc}
\vspace{-2mm}
\centering
\resizebox{0.9\textwidth}{!}{%
\centering
\begin{tabular}{|l|l|l|lllllllll|}
\hline
 \multicolumn{3}{|c|}{} & \multicolumn{9}{c|}{Modulus Accuracy (TPR)} \\
 \multicolumn{3}{|c|}{1000 Iterations} & 1.0 & 0.95 & 0.9 & 0.85 & 0.8 & 0.75 & 0.7 & 0.65 & 0.6 \\\cline{1-12}
\multirow{14}{*}{\rotatebox[origin=c]{90}{Non-Modulus Accuracy (TNR)}} & \multirow{14}{*}{\rotatebox[origin=c]{90}{BehavioCog (BC)}} & 1.0 & 165.978 & 262.996 & 345.960 & 451.254 & 565.272 & 678.194 & 823.412 & 982.746 & 1235.992 \\
 &  & 0.95 & 174.384 & 280.342 & 385.178 & 488.884 & 612.682 & 766.068 & 922.540 & 1120.090 & 1342.758 \\
 &  & 0.9 & 181.922 & 294.812 & 416.188 & 524.556 & 676.894 & 846.600 & 1025.982 & 1286.478 & 1580.582 \\
 &  & 0.85 & 193.792 & 320.708 & 459.760 & 594.082 & 776.850 & 971.874 & 1178.938 & 1483.118 & 1891.290 \\
 &  & 0.8 & 207.030 & 355.722 & 508.420 & 675.650 & 867.420 & 1077.646 & 1377.088 & 1761.734 & 2242.302 \\
 &  & 0.75 & 216.218 & 381.012 & 549.340 & 730.038 & 960.940 & 1227.158 & 1629.738 & 2094.498 & 2808.506 \\
 &  & 0.7 & 234.262 & 440.108 & 624.560 & 836.196 & 1115.732 & 1476.358 & 1949.384 & 2522.530 & 3340.990 \\
 &  & 0.65 & 258.334 & 475.660 & 685.860 & 989.504 & 1331.370 & 1754.882 & 2372.312 & 3243.770 & 4297.246 \\
 &  & 0.6 & 279.748 & 522.326 & 803.298 & 1157.238 & 1586.432 & 2203.824 & 2940.524 & 4112.454 & 5987.648 \\ 
  &  & 0.55 & 303.129 & 616.505 & 948.832 & 1348.100 & 1911.934 & 2778.192 & 3909.067 & 5627.870 & 8521.225 \\
 &  & 0.5 & 330.702 & 705.391 & 1126.371 & 1661.984 & 2392.133 & 3593.075 & 5396.283 & 8248.908 & 13557.03 \\
 &  & 0.45 & 369.793 & 829.916 & 1344.850 & 2093.507 & 3149.585 & 5070.656 & 7984.902 & 13339.76 & 25016.93 \\
 &  & 0.4 & 414.546 & 996.672 & 1683.427 & 2815.620 & 4455.89 & 7358.926 & 13297.78 & 25743.05 & 63722.74 \\
 &  & 0.35 & 473.636 & 1234.009 & 2218.427 & 3832.711 & 6882.528 & 12585.91 & 26499.93 & 74077.08 & >200000 \\
 \hline
\end{tabular}%
}
\vspace{3mm}
\end{table*}

\begin{table*}[t]
\begin{minipage}[t]{.74\textwidth}
\raggedright
\caption{Experimentally derived rounds required to reveal full user secret given varying TPR and TNR of side-channel classifier, for the Modulus applied on FoxTail (16,290 Round Benchmark\cite{asghar2015linearization}), and HopperBlum.}
\label{tab:mod_algo_results}
\vspace{-2mm}
\resizebox{\textwidth}{!}{%
\centering
\begin{tabular}{|l|l|l|lllllllll|}
\hline
 \multicolumn{3}{|c|}{} & \multicolumn{9}{c|}{Modulus Accuracy (TPR)} \\
 \multicolumn{3}{|c|}{1000 Iterations} & 1.0 & 0.95 & 0.9 & 0.85 & 0.8 & 0.75 & 0.7 & 0.65 & 0.6 \\\cline{1-12}
\multirow{26}{*}{\rotatebox[origin=c]{90}{Non-Modulus Accuracy (TNR)}} 
 & \multirow{13}{*}{\rotatebox[origin=c]{90}{FoxTail (FT)}} & 1.0 & 234.099 & 361.049 & 462.091 & 590.253 & 716.030 & 851.755 & 996.776 & 1206.307 & 1431.352 \\
 &  & 0.95 & 245.924 & 384.369 & 512.155 & 631.904 & 779.754 & 929.230 & 1132.713 & 1325.878 & 1574.523 \\
 &  & 0.9 & 259.346 & 418.701 & 555.846 & 696.247 & 846.392 & 1039.581 & 1275.197 & 1514.627 & 1826.617 \\
 &  & 0.85 & 275.493 & 450.995 & 591.611 & 751.928 & 945.971 & 1155.384 & 1412.938 & 1694.107 & 2061.644 \\
 &  & 0.8 & 289.746 & 474.639 & 657.124 & 846.583 & 1073.885 & 1341.471 & 1625.042 & 2034.259 & 2407.814 \\
 &  & 0.75 & 309.022 & 518.628 & 719.507 & 942.644 & 1184.645 & 1527.110 & 1873.568 & 2395.379 & 2942.497 \\
 &  & 0.7 & 337.233 & 570.886 & 807.810 & 1076.267 & 1376.793 & 1780.750 & 2223.255 & 2780.624 & 3536.249 \\
 &  & 0.65 & 360.027 & 654.276 & 903.714 & 1222.030 & 1561.489 & 2110.796 & 2593.069 & 3363.440 & 4372.477 \\
 &  & 0.6 & 390.420 & 719.669 & 1027.850 & 1420.807 & 1857.157 & 2491.812 & 3303.234 & 4271.355 & 5663.664 \\ 
 &  & 0.55 & 425.797 & 817.616 & 1211.947 & 1673.944 & 2309.427 & 3092.693	& 4160.624	& 5502.299	& 7751.852 \\
 &  & 0.5 & 466.579 & 928.206 & 1432.573 & 2044.601 & 2804.419 & 3836.378	& 5499.969	& 7662.418	& 11093.98 \\ 
 &  & 0.45 & 519.760 & 1112.707 & 1724.578 & 2495.029 & 3612.860 & 5179.980	& 7484.748	& 11356.39	& 17528.47 \\
 &  & 0.4 & 581.398 & 1302.668 & 2123.358 & 3243.961 & 4803.101 & 7302.969	& 11542.69	& 18748.54	& 33004.62 \\
 & & 0.35 & 665.757	& 1592.206	& 2786.819	& 4429.673	& 7045.669 & 	11521.79	& 19765.98	& 38068.37	& 83833.69 \\ \cline{2-12}
 & \multirow{13}{*}{\rotatebox[origin=c]{90}{HopperBlum (HB)}} & 1.0 & 538.108 & 643.254 & 759.623 & 912.669 & 1056.444 & 1233.634 & 1448.241 & 1702.210 & 2011.441 \\ 
 &  & 0.95 & 574.626 & 684.327 & 809.546 & 972.457 & 1144.787 & 1350.292 & 1616.565 & 1896.063 & 2251.269 \\
 &  & 0.9 & 601.860 & 734.153 & 871.632 & 1041.118 & 1280.526 & 1502.996 & 1810.434 & 2120.654 & 2562.250 \\
 &  & 0.85 & 636.235 & 780.374 & 957.234 & 1137.184 & 1364.196 & 1663.844 & 2056.685 & 2485.285 & 3047.063 \\
 &  & 0.8 & 674.237 & 833.574 & 1031.503 & 1271.620 & 1537.344 & 1922.406 & 2356.007 & 2886.691 & 3512.182 \\
 &  & 0.75 & 712.093 & 908.231 & 1129.152 & 1388.736 & 1757.375 & 2152.616 & 2687.703 & 3351.902 & 4257.354 \\
 &  & 0.7 & 784.420 & 1000.224 & 1230.416 & 1580.273 & 1964.927 & 2466.244 & 3146.388 & 4113.707 & 5272.378 \\
 &  & 0.65 & 852.850 & 1066.540 & 1383.679 & 1789.643 & 2253.984 & 2906.097 & 3840.653 & 4946.336 & 6689.905 \\
 &  & 0.6 & 908.848 & 1181.096 & 1574.578 & 2081.068 & 2654.442 & 3515.377 & 4805.121 & 6448.616 & 8940.049 \\ 
 &  & 0.55 & 983.286 & 1333.104 & 1812.558 & 2402.852 & 3259.621 & 4408.826	& 5955.009	& 8593.798	& 12597.03 \\
 &  & 0.5 & 1087.780 & 1495.170 & 2119.745 & 2964.009 & 4068.224 & 5571.716	& 8035.144	& 12339.58	& 19706.51 \\
 &  & 0.45 & 1225.173 & 1763.995 & 2510.148 & 3626.290 & 5112.787 & 7750.267	& 11859.26	& 19966.72	& 35981.45 \\
 &  & 0.4 & 1356.933 & 2056.374 & 3105.767 & 4595.770 & 7067.243 & 11288.44	& 19362.25	& 37558.86	& 91942.35 \\ 
 & & 0.35 & 1554.770	& 2544.067	& 3914.772	& 6296.296	& 10456.32 & 	18819.72	& 38944.80	& 102742.2	& >200000 \\ \hline
\end{tabular}%
}
\end{minipage}%
\vspace{2mm}
\hfill
\begin{minipage}[t]{.24\textwidth}
\raggedleft
\caption{Point update for BC, FT, HB. A cell is divided into a upper and lower half, representing the detection of a modulus and no-modulus respectively.}
\label{tab:mod_point_update}
\vspace{-1mm}
\begin{subtable}{1\textwidth}
\raggedleft
\resizebox{0.9\textwidth}{!}{%
\begin{tabular}{| c | c | c |c |c |c |}
\multicolumn{6}{c}{a) BC} \\
\hline
\diagbox[innerrightsep=1pt, innerleftsep=1pt]{$r$}{$w$} & 0 & 1 & 2 & 3 & 4 \\ \hline
0 & \ourdiagbox{0}{\hphantom{-}0}& \ourdiagbox{0}{-1}& \ourdiagbox{0}{-1}& \ourdiagbox{0}{-1}& \ourdiagbox{0}{-1}\\ \hline
1 & \ourdiagbox{0}{\hphantom{-}0}& \ourdiagbox{0}{\hphantom{-}0}& \ourdiagbox{0}{-1}& \ourdiagbox{0}{-1}& \ourdiagbox{0}{-1}\\ \hline
2 & \ourdiagbox{0}{\hphantom{-}0}& \ourdiagbox{0}{\hphantom{-}0}& \ourdiagbox{0}{\hphantom{-}0}& \ourdiagbox{0}{-1}& \ourdiagbox{0}{-1}\\ \hline
3 & \ourdiagbox{0}{\hphantom{-}0}& \ourdiagbox{0}{\hphantom{-}0}& \ourdiagbox{0}{\hphantom{-}0}& \ourdiagbox{0}{\hphantom{-}0}& \ourdiagbox{0}{-1}\\ \hline
4 & \ourdiagbox{0}{\hphantom{-}0}& \ourdiagbox{0}{\hphantom{-}0}& \ourdiagbox{0}{\hphantom{-}0}& \ourdiagbox{0}{\hphantom{-}0}& \ourdiagbox{0}{\hphantom{-}0}\\ \hline
\end{tabular}
}
\noindent \\
\vspace{3mm}
\resizebox{0.95\columnwidth}{!}{%
\begin{tabular}{l | c | c | c |c |c |}
\cline{2-6}
& \diagbox[innerrightsep=1pt, innerleftsep=1pt]{$r$}{$w$} & 0 & 1 & 2 & 3\\ \cline{2-6}
b) FT & 0 & \ourdiagbox{0}{\hphantom{-}0}& \ourdiagbox{0}{\hphantom{-}0}& \ourdiagbox{0}{-1}& \ourdiagbox{0}{-1}\\ \cline{2-6}
& 1 & \ourdiagbox{0}{\hphantom{-}0}& \ourdiagbox{0}{\hphantom{-}0}& \ourdiagbox{0}{\hphantom{-}0}& \ourdiagbox{0}{\hphantom{-}0}\\ \cline{2-6}
\end{tabular}
}
\noindent \\
\vspace{3mm}
\resizebox{0.66\columnwidth}{!}{%
\raggedleft
\begin{tabular}{l | c | c | c |}
\cline{2-4}
& \diagbox[innerrightsep=1pt, innerleftsep=1pt]{$r$}{$w$} & 0 & 1\\ \cline{2-4}
c) HB & 0 & \ourdiagbox{0}{\hphantom{-}0}& \ourdiagbox{0}{-1}\\ \cline{2-4}
& 1 & \ourdiagbox{0}{\hphantom{-}0}& \ourdiagbox{0}{\hphantom{-}0}\\ \cline{2-4}
\end{tabular}
}
\end{subtable}%
\end{minipage}
\vspace{2mm}
\end{table*}%

Obtaining an analytical estimate of the number of samples required to retrieve the secret through the algorithm is difficult. We therefore assess this through simulations. The penalty vectors chosen satisfy the condition of the theorem above. In particular, we use the penalty vector $(u_0, u_1, \ldots, u_{d-1}) = (-1, -1, \ldots, -1)$. With this penalty vector, the point update follows the pattern shown in Table~\ref{tab:mod_point_update}, for each of the three schemes. We varied the TPR and TNR of the oracle between 0.6 and 1.0 with steps of 0.05. We maintain the distinction between TPR and TNR to preserve the asymmetrical effects on our algorithm resulting from each type of error. For each pair of TPR and TNR we ran 1,000 simulations of the attack algorithm on each of the three schemes. Instead of giving $m$, i.e., the number of challenges, as an algorithm input, we let it run until the top $k$ items are the secret items. Tables \ref{tab:mod_algo_results_bc} and \ref{tab:mod_algo_results} contains the average number of rounds required for all schemes.

We compare our results from Tables~\ref{tab:mod_algo_results_bc} and Table~\ref{tab:mod_algo_results} to the samples required by best performing efficient algebraic attacks. For BehavioCog (in Table~\ref{tab:mod_algo_results_bc}), the most efficient attack is Gaussian elimination, which finds the secret in 900 rounds~\cite{chauhan2017behaviocog}. For the FoxTail protocol (In Table~\ref{tab:mod_algo_results}), linearization followed by Gaussian elimination requires 16,290 observations~\cite{asghar2015linearization}. In terms of number of sessions, this is 90 sessions for BehavoCog (10 rounds per session) and 815 for FoxTail (20 rounds per session). In comparison, there are various ranges of accuracy levels for the faulty oracle which reduce the average number of sessions required to obtain the secret. Taking a realistic example of (TPR, TNR) = $(1.0, 0.6)$, BehavioCog would only need 279.7 rounds (28 sessions), whilst FoxTail would need 390.4 rounds (20 sessions), a substantially lower number of complete authentication sessions. For the HB protocol (In Table~\ref{tab:mod_algo_results}), no known efficient algebraic attack exists. This, is not surprising as the protocol is based on the NP-Hard problem of learning parity with noise. For HB protocol, we require 909 rounds or approximately 27 sessions with 34 rounds per session (as discussed in Section~\ref{sec:existing_protocols}). Thus, our attack shows an efficient attack based on side channel information.\footnote{We note that Cagalj et al.~\cite{vcagalj2015timing} propose a side channel attack on the original HB protocol (without the window) which recovers the secret in 380 rounds. Unfortunately this attack is not applicable to the windowed variant.} 

In comparison, Tables~\ref{tab:mod_algo_results_bc} and \ref{tab:mod_algo_results} shows that our attack algorithm with many different combinations of TPR-TNR far outperforms the aforementioned attacks. But, which combinations of TPR-TNR are fair and realistic? We see that if the TPR is high (close to 1), then our algorithm is less sensitive to decreasing TNR. This is somewhat evident from Algorithm~\ref{algo:mod_point_update}, and our choice of the penalty vector. A low TNR means that a no-modulus event may be frequently misclassified as a modulus event. However, the algorithm never penalizes items in such a case (due to the use of a zero penalty vector). This is also due to the fact that in a no-modulus event, there is a large weight differential (as shown previously). For binary classification tasks, it is always possible to trade the TNR with diminishing FPR, as TPR and TNR are inversely related. The first few columns in the table correspond to this regime, and we see that in most cases we significantly outperform algebraic attacks.  

On the other hand, in the case of BehavioCog, if both TPR and TNR are less than 0.8 then the performance of the side-channel classifier degrades in comparison to Gaussian elimination. However, in general, the classifier can be trained to favor the TNR over TPR, or vice versa, by varying the threshold. This means that we can fix the threshold to favor a high TPR, say 0.95 or 1.0, sacrificing the TNR as a result but still be able to outperform Gaussian elimination in terms of the number of rounds required to retrieve the secret, as is evident from the table (the first two columns under BC). Indeed, as we shall show in Section~\ref{sec:mod_classifier}, we are able to train the classifier on data from our user study to achieve a TPR of 1.0 and a TNR of 0.38, and still able to obtain the secret after 435 rounds, less than half the number required via Gaussian elimination.

\paragraph*{Increasing Confidence}
The results in Tables~\ref{tab:mod_algo_results_bc} and \ref{tab:mod_algo_results} show the average minimum number of rounds required before the first $k$ items are the user's secret items. While, these numbers can be used as a reference on how many challenge-response pairs are required to find the secret via Algorithm~\ref{algo:mod_point_update} with high probability, the attacker can use a better strategy to increase its confidence on the first $k$ items being the secret items. The idea is to rank items after each round according to their scores, and keeping track of the point differences between the neighbors of the $k$th ranked item. More precisely, let $\text{item}(i)$ denote the item ranked $i$ in the current round, where $1 \le i \le n$ (ties can be broken according to the initial order on the items). Note that the item ranked $i$ might change over successive rounds. The attacker updates the (absolute) difference in points of the following pairs of items: $(\text{item}(k - 1), \text{item}(k))$, $(\text{item}(k), \text{item}(k+1))$, and $(\text{item}(k + 1), \text{item}(k + 2))$. Let us denote these three point differences by $\text{diff}_{k-1, k}$, $\text{diff}_{k, k+1}$ and $\text{diff}_{k+1, k+2}$, respectively. For the first few rounds, the attacker cannot distinguish between these three. For a given (TPR, TNR), once the number of rounds passes the mark given in Tables~\ref{tab:mod_algo_results_bc} and \ref{tab:mod_algo_results}, $\text{diff}_{k, k+1}$ starts deviating away from $\text{diff}_{k-1, k}$ and $\text{diff}_{k+1, k+2}$. The more rounds the attacker observes, the more $\text{diff}_{k, k+1}$ deviates away from the two. Thus, the attacker can increase its confidence that the top $k$ are indeed the secret items by setting a threshold for the gap between $\text{diff}_{k-1, k}$ with respect to $\text{diff}_{k-1, k}$ and $\text{diff}_{k+1, k+2}$. Figure~\ref{fig:increase_conf} shows this for BehavioCog with $(\text{TPR}, \text{TNR}) =  (0.95, 0.95)$. We can see a divergence in the score differences after around 280 rounds, consistent with our simulated rounds required for this configuration (cf. Table~\ref{tab:mod_algo_results_bc}).

\begin{figure}[t]
	\centering
    \includegraphics[width=0.6\columnwidth]{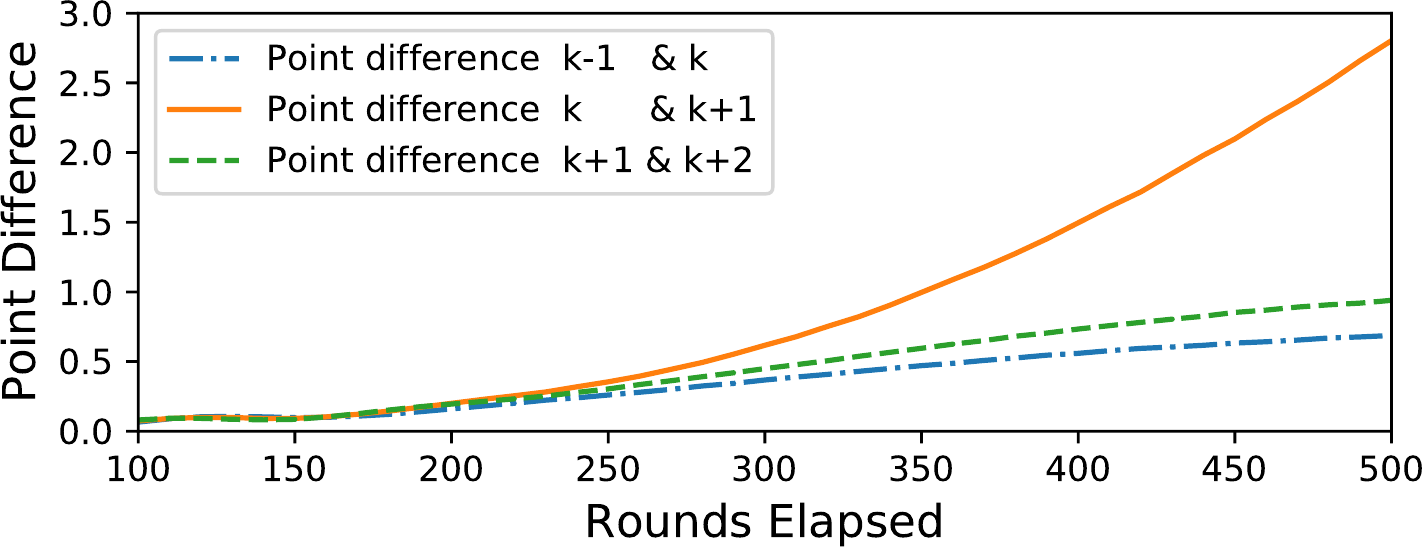}
    \vspace{-2mm}
  	\caption{The point difference between the $k$th and $(k+1)$st ranked items versus the points difference between $(k-1)$st and $k$th ranked, and $(k+1)$st and $(k+2)$nd ranked items as a function of number of observed rounds. These results are for a faulty oracle with 0.95 EER on BehavioCog. Clearly, after around 280 (expected number of rounds to find the secret), there is growing divergence between the scores of secret and decoy items, indicating increased confidence in the top $k$ items being the secret items.
  	}\vspace{2mm}
 	\label{fig:increase_conf}
\end{figure}

\section{Implementing the Attack Using Behavioral Side-Channel}
\label{sec:classifier}
In this section, we show that certain user behavior patterns while processing challenges can provide information about the modulus event. More specifically, we target the user's eye-movement together with the associated timing information. An eavesdropping adversary's ability to accurately guess the modulus event depends on the resolution of behavioral information available. We model this as adversaries with varying strength (Section~\ref{sub:adv-levels}); from the weakest adversary with access to only meta information to the strongest adversary with high-\-resolution eye movement to screen mapping.
We then identify potentially revealing behavior patterns through a user study by collecting data from an eye-tracker (Section~\ref{sec:userdata}). Following this, we identify features corresponding to these behavior patterns which are then used as input to machine learning classifiers (Section~\ref{sec:features}) to predict the modulus event (thus instantiating the faulty oracle of the previous section). The data from the user study is used to train and test the classifiers, and the resulting accuracy levels (TPR and TNR) are used as instances of the faulty oracle in the aforementioned attack algorithm.

\subsection{Levels of Adversarial Strength}
\label{sub:adv-levels}
We define four different levels of adversaries differing by the resolution of behavioral information available to them. 
These levels are outlined below with real-world examples. Figure~\ref{fig:adversary_levels} illustrates them pictorially. We assume each adversary can access the challenge and responses in addition to the behavioral information. 
\begin{itemize}
\item \emph{Level 1 (L1):} An adversary with access to the challenge duration, i.e., time till user submits response. Examples include monitoring Internet traffic or the screen itself.

\item \emph{Level 2 (L2):} An adversary with further access to user dwell times, i.e., when the eye is stationary. This information can be obtained via a hidden camera facing the user, e.g., a pinhole camera mounted on an ATM, or a general surveillance camera. The resulting video feed of the user's eyes can be used to determine still positions through pupil detection and its lack of movement.
    
\item \emph{Level 3 (L3):} An adversary with further access to rudimentary positional information of dwells, e.g., lower half of the screen, top-left quadrant. This information can again be obtained via a hidden camera recording the user's face. Furthermore, we assume that the attacker has access to video-oculography to estimate gaze and to extract positional information from either the geometric model or appearance of the eyes~\cite{krafka2016eye, hansen2010eye, baltruvsaitis2016openface, dalmaijer2014pygaze}.
    
\item \emph{Level 4 (L4):} An adversary with further access to item-specific positional information of dwells. We assume the attacker employs a hidden camera to record a video of the user's face. The attacker has access to highly accurate video-oculography~\cite{krafka2016eye, hansen2010eye, baltruvsaitis2016openface, dalmaijer2014pygaze} to estimate item specific positional information (as compared to coarse-grained positions in L3).  
\end{itemize}

\paragraph*{Note:} A webcam or the front camera of a smartphone or a laptop, are also possible examples of a hidden camera considered for adversary levels 2 to 4. However, it can be argued that the user's device is already compromised if an attacker has access to the in-device camera, and hence the protection provided by an ORAS might be superfluous. Therefore, we discard this as a possible attack vector, and instead consider off-device hidden cameras, examples of which are given above.

\begin{figure*}[t]
\begin{minipage}[t]{.74\textwidth}
\raggedright
    \includegraphics[width=1.0\textwidth]{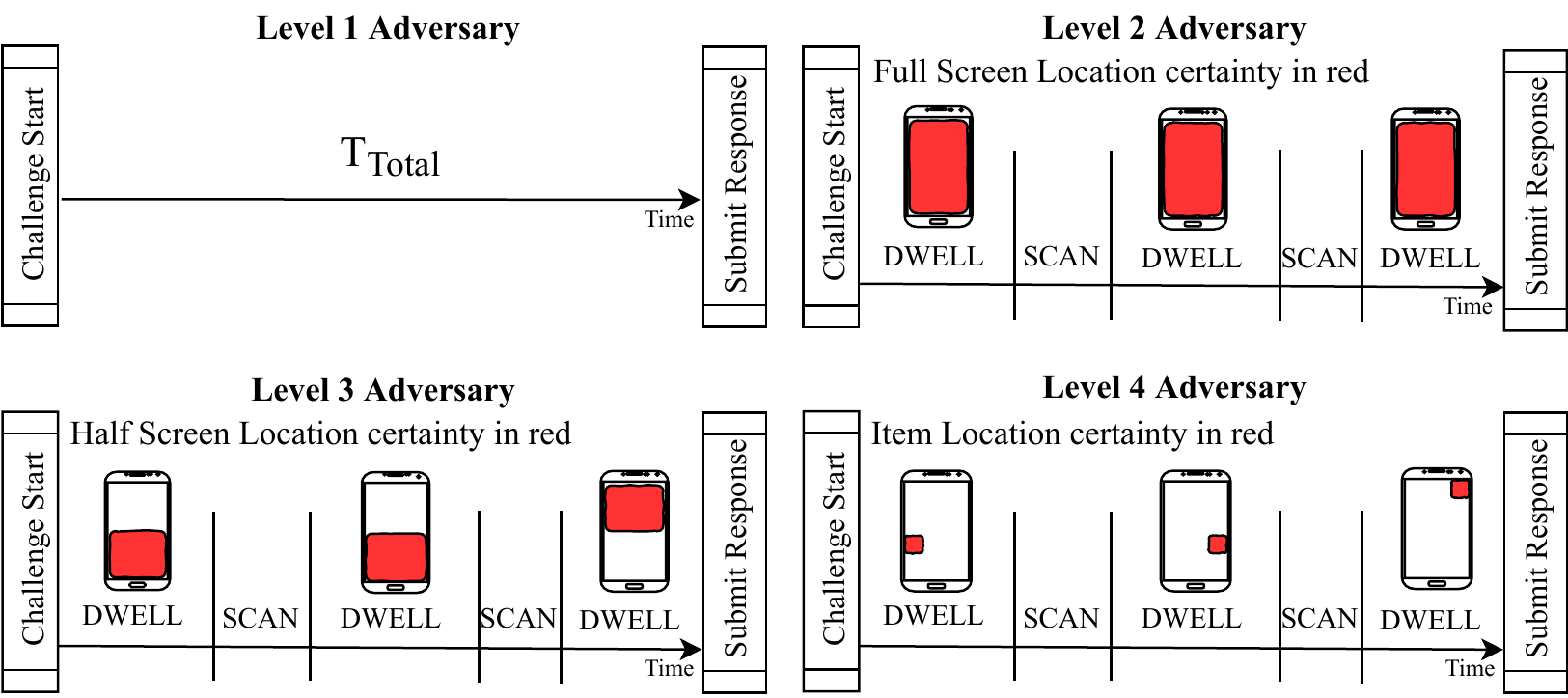}
    \vspace{-2mm}
  	\caption{Four Levels of adversary capabilities in recovering eye-tracking information, Each level beyond L1 is provided with increasingly detailed location information, from no location information (L2), sectors (L3), to specific items (L4).}
    \label{fig:adversary_levels}
\end{minipage}%
\vspace{2mm}
\hfill
\begin{minipage}[t]{.24\textwidth}
\raggedleft
    \includegraphics[width=1.0\columnwidth]{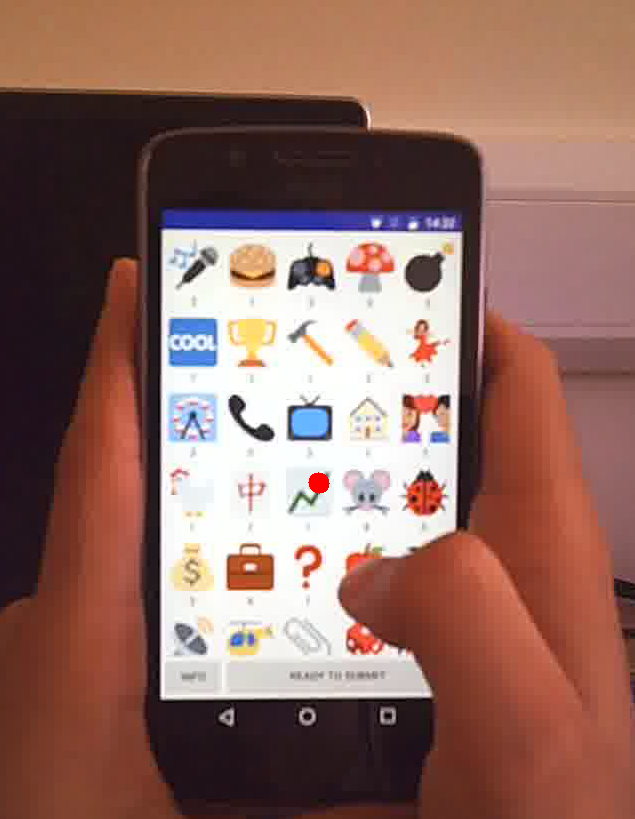}
    \vspace{-2mm}
  	\caption{An image of the $xy$-coordinate (red dot) overlaid on a video feed.}
 	\label{fig:xy_image_mapping}
\end{minipage}
\vspace{2mm}
\end{figure*}%

\subsection{User Study}
\label{sec:userdata}
We recruited 11 postgraduate research students, in the 24-26 age range of mixed gender (6 males, 5 females) as participants in the eye-tracking experiment who were asked to process challenges from the BehavioCog scheme (specifically, the cognitive component in \cite{chauhan2017behaviocog}) with parameters $(n, l, k, d) = (180, 30, 14, 5)$.
The users were given training, and trial attempts to help them remember their secrets and to familiarize themselves with the scheme, followed by computing several random challenges. The position of their gaze, and field of view are recorded with a pair of SMI Eye Tracking Glasses (ETG2).\footnote{\url{https://www.smivision.com/eye-tracking/products/mobile-eye-tracking}}
From all challenges attempted by the users, we sampled those that had correct responses such that there were roughly the same number of instances of $g \in \{0,1,2,3\}$ secret items. This represented 81.4\% of possible challenges within BehavioCog, resulting in a total of 64 challenge samples. A full breakdown of user samples with respect to the number of secrets present $g$, and whether they had performed a modulus operation is shown in Table \ref{tab:user_breakdown}.

\begin{table}[t!]
\centering
\caption{User Contribution of Eye-tracking Samples}
\label{tab:user_breakdown}
\vspace{-2mm}
\resizebox{0.64\columnwidth}{!}{%
\begin{tabular}{| ll | x{1.5mm}x{1.5mm}x{1.5mm}x{1.5mm}x{1.5mm}x{1.5mm}x{1.5mm}x{1.5mm}x{1.5mm}x{3mm}x{3mm} | x{8mm} |}
\hline
64 Samples & User & 1&2&3&4&5&6&7&8&9&10&11&Total \\
\hline
\hline
\multirow{4}{*}{\shortstack[l]{Number\\of\\Secrets}} &  
  0 Secrets & 2 & 1 & 1 & 0 & 1 & 3 & 2 & 1 & 0 & 2 & 1 & 14 \\
& 1 Secrets & 1 & 2 & 2 & 1 & 2 & 1 & 1 & 1 & 0 & 2 & 2 & 15 \\
& 2 Secrets & 2 & 2 & 2 & 2 & 0 & 2 & 2 & 1 & 2 & 1 & 2 & 18 \\
& 3 Secrets & 2 & 1 & 1 & 3 & 0 & 1 & 2 & 0 & 3 & 3 & 1 & 17 \\
\hline \hline
\multirow{3}{*}{\shortstack[l]{Modulus\\Status}} &
   Empty  & 2 & 1 & 1 & 0 & 1 & 3 & 2 & 1 & 0 & 2 & 1 & 14 \\
 & No Mod & 3 & 4 & 4 & 3 & 2 & 3 & 5 & 2 & 3 & 3 & 4 & 36 \\
 & Mod    & 2 & 1 & 1 & 3 & 0 & 1 & 0 & 0 & 2 & 3 & 1 & 14\\
\hline
\end{tabular}
}
\vspace{2mm}
\end{table}

The data from the study included (a) a recorded video of each challenge, (b) and an associated list of $xy$-coordinates. We performed a manual mapping process to link an $xy$-coordinate to a challenge item in the video.
The task was partially automated by overlaying the timestamped $xy$-coordinate as a red dot on the corresponding frame of the video feed, as shown in Figure~\ref{fig:xy_image_mapping}.
The SMI eyetracker software provides its own classification of eye movement as either ``Saccade'' (rapid eye movement or scan), ``Visual Intake'' (low eye movement or dwell), and ``Blink.'' The mapped item over which the user's focus is positioned corresponds to segments of visual intake separated by saccades or blinks. 
We noticed that the SMI software classification is highly sensitive to even the smallest eye movements, whereby if a user shifts their focus on different parts of the same item, a saccade may be registered. An attacker may not have the same luxury of information. So we treat sequential periods of visual intake separated by a saccade or blink of the same item as a continuous dwell on an item. The mapping of positional information to challenge items enables us to produce detailed information available to L3 and L4 adversary. However for less capable adversaries (L1 and L2) such positional information is not required.

\subsubsection{Ethics Consideration}
The participants were recruited via university mailing lists and posters, and were informed about why and how their data was to be used. Their consent for the collection of eye-tracking data was obtained, with monetary compensation provided to the participant at the completion of the experiment. Ethics approval for the conduct of the experiment and analysis of the data was obtained from our ethics review board prior to user recruitment. The recordings of our participant interactions may contain potentially identifiable information (fingerprints, skin tone), thus after the mapping of focal $xy$-coordinates and challenge items, the recordings were encrypted at rest.

\subsection{Features and Classifiers}
\label{sec:features}
From the eye-tracker study, we identified several features, e.g., minimum dwell time, number of vertical transitions (between the two screen halves). Features are categorized according to their availability to the four adversarial levels. Justifications and the hypotheses behind the choice of these features, i.e., how they intuitively reveal information about the modulus/no-modulus event, are included in Appendix~\ref{appendix:features}.
The features are also listed in Table~\ref{tab:order_feat}. Due to the small sample size, providing all features to a machine learning classifier is not recommended (due to the curse of dimensionality~\cite{goodfellow-deep}). Thus features were ranked using the minimal-redundancy-maximal-relevance (mRMR) score~\cite{peng2005feature}. The mRMR algorithm seeks to rank features to maximize the information gain provided by a feature for the task of separating the sample classes. The algorithm also accounts for redundant features; otherwise multiple similar features would be ranked highly, ignoring the fact that each subsequent feature would provide little new information. Table~\ref{tab:order_feat} details the feature rankings in ascending order of information gained.

\begin{table}
    \caption{MRMR Rankings for features of modulus side channel classification.}
    \label{tab:order_feat}
    \vspace{-2mm}
    \centering
\resizebox{0.55\columnwidth}{!}{%
    \begin{tabular}{|l|l|l|rrrr|}\cline{3-7}
    \multicolumn{2}{c|}{} &  \textbf{Adversary Level}& L1 & L2 & L3 & L4 \\ \cline{1-7}
    \multirow{25}{*}{\rotatebox[origin=c]{90}{\textbf{Engineered Features}}} & \multirow{3}{*}{\rotatebox[origin=c]{90}{Level 1}} & Total Time of authentication & 1 & 1 & 1 & 1 \\
    & & Mean Challenge Weight & 2 & 3 & 3 & 3 \\
    & & User Challenge Response & 3 & 5 & 5 & 5 \\ \cline{2-7}
    & \multirow{12}{*}{\rotatebox[origin=c]{90}{Level 2}}& Minimum Dwell Time &  & 6 & 7 & 11 \\ 
    & & 10\textsuperscript{th} Percentile Dwell Time &  & 8 & 10 & 16 \\
    & & Maximum Dwell Time &  & 10 & 14 & 20 \\
    & & 90\textsuperscript{th} Percentile Dwell Time &  & 12 & 16 & 22 \\
    & & Mean Dwell Time &  & 11 & 15 & 21 \\
    & & STD Dwell Time &  & 15 & 19 & 25 \\
    & & Number of Dwells &  & 13 & 17 & 23 \\
    & & Time to end from longest Dwell &  & 14 & 18 & 24 \\
    & & Dwell Consistency &  & 2 & 2 & 2 \\
    & & Duration of First Fixation &  & 7 & 9 & 15 \\
    & & Duration of Last Fixation &  & 9 & 12 & 18 \\
    & & Longest Dwell Consistency &  & 4 & 4 & 4 \\ \cline{2-7}
    & \multirow{4}{*}{\rotatebox[origin=c]{90}{Level 3}}& Vertical Transitions (Even) &  &  & 6 & 10 \\
    & & Vertical Transitions &  &  & 8 & 13 \\
    & & Horizontal Transitions &  &  & 11 & 17 \\
    & & Time to end from screen bottom &  &  & 13 & 19 \\ \cline{2-7}
    & \multirow{6}{*}{\rotatebox[origin=c]{90}{Level 4}}& Number of largest revisits &  &  &  & 8 \\
    & & Number of unvisited items &  &  &  & 14 \\
    & & Longest repeating sequence &  &  &  & 7 \\
    & & Weight of 1\textsuperscript{st} Longest Dwell &  &  &  & 6 \\
    & & Weight of 2\textsuperscript{nd} Longest Dwell &  &  &  & 9 \\
    & & Weight of 3\textsuperscript{rd} Longest Dwell &  &  &  & 12 \\ \hline
    \end{tabular}
}
\vspace{2mm}
\end{table}

A comprehensive selection of classification algorithms was tested from the Python machine learning library \texttt{scikit-learn}~\cite{scikit-learn}. Specifically, we use Support Vector Machines (linear and radial kernel), Naive Bayes, AdaBoost and Random Forest classifiers. Each algorithm was tested on an increasing number of features for each adversary level, as determined by the mRMR algorithm. Each of our classifiers are used in a two-class configuration for the modulus/no-modulus event.

We adopt leave-one-user-out verification as the most rigorous form of model validation, allowing the demonstration of generic behaviors irrespective of user. The method is a proactive assurance against overfitting; with the low number of available training samples, the inclusion of any user specific samples would risk the trained model learning user-specific behavior instead of generic behaviors across the entire group of users. This also represents a realistic attack scenario whereby the attacker has no prior knowledge of the target user.

As previously observed in our simulations (Table~\ref{tab:mod_algo_results}), the performance of the attack algorithms is disproportionately sensitive to the accuracy of detecting one class over the other (modulus event accuracies are more important than the no-modulus accuracy). Each of the classifiers return a prediction probability score for each class label. By default, a threshold is set to 50\%, a sample is classified as belonging to the first class if the score returned by the classifier is 50\% or above. We can favor either class by altering this threshold to tighten or loosen the conditions for being classified into the first class, thus controlling the trade-off between TPR and TNR.

The best performing classifier (algorithm, features, threshold), is then found by using the TPR and TNR values in the faulty oracle in the points update algorithm and simulating challenge-response rounds in the scheme. The simulation is repeated 1,000 times to obtain an average number of rounds. The classifier with the lowest number of rounds to resolve the secret, is chosen. This process is then repeated for all adversarial levels. We note that a single global threshold is used across every testing sample, irrespective of the validation fold. To obtain a single value of the pair (TPR, TNR), we aggregate the test samples from each fold into a set.
We acknowledge the low number of test samples prevents us from directly attacking a user, instead having to adopt challenge-response simulations. Additionally with a larger test group of users, more data can be leveraged to train better performing machine learning models.

\begin{table}[t]
\centering
\caption{Best adversary level classifier exploitation of the modulus operation information.}
\label{tab:mod_classifier_algo_results}
\vspace{-2mm}
\resizebox{0.58\columnwidth}{!}{%
\begin{tabular}{| l | ll | l | l | l |}
\hline
Adver.	&No-Mod  &Mod   &Rounds & Classifier & \# Features, \\ 
Level	&Acc.  &Acc.    &required & used & Threshold\\
\hline
L1   & 0.38	&1.00   & 435.04 &Adaboost & 1, ~0.51\\
L2   & 0.38	&1.00   & 435.04 &Adaboost & 1, ~0.51\\
L3   & 0.38	&1.00   & 435.04 &Adaboost & 1, ~0.51\\
L4   & 0.40	&1.00   & 411.89 &Naive Bayes & 7, ~0.59\\
\hline
\end{tabular}
}
\vspace{2mm}
\end{table}

\subsection{Modulus Event Side Channel}
\label{sec:mod_classifier}

After training and testing classifiers on the ranked features, the results of the simulations with the classifiers as faulty oracles (given by corresponding TPR and TNR) are presented in Table~\ref{tab:mod_classifier_algo_results}. With only one feature, i.e., ``Total Time'', the AdaBoost classifier was able to obtain a (no-mod, mod) accuracy of (0.38, 1.0). Unfortunately, the additional features provided to the classifier in adversarial levels L2 and L3 did not improve our algorithm performance further; Until L4, where the naive Bayes classifier is able to achieve an accuracy of (0.4, 1.0) with 7 features. Note that we did not change the threshold over the default 0.5 by a large degree to obtain a TPR of 1.0. The results for the higher level adversaries do not show a significant improvement over lower level adversaries. But this may be due to our limited field study. Since these accuracy levels are dependent on the data from the user study, a larger user study might reflect better on the influence of other features in classification accuracy. On the other hand, a Level 1 adversary with only the total time of the challenge can sufficiently separate the modulus and the non-modulus challenges, demonstrating the practicality of our attack and the need to consider user behavior when designing ORAS.

With these oracle accuracies, our simulations show that it will take approximately 435 observations on average for a L1-L3 adversary, and 412 for L4 adversary to find the user's secret. This is half of the rounds needed by the Gaussian elimination attack (900 rounds) in the BehavioCog scheme~\cite{chauhan2017behaviocog}. By extending these simulations to the FoxTail and HopperBlum schemes, we observe 589 and 1,346 rounds, respectively, for an L4 adversary, and 618 and 1,415 rounds, respectively, for L1-L3 adversaries. Recall that the linearisation/Gaussian elimination attack on Foxtail requires 16,290 rounds, whereas the HB protocol has no efficient algebraic or statistical attack. 

Finally, comparing the number of rounds in Table~\ref{tab:mod_classifier_algo_results} against the numbers reported in Table~\ref{tab:mod_algo_results}, we see that the number of rounds required by the best classifier via the user study is larger than the simulated attacks. However, we reiterate that this is due to the best accuracy level through our limited user study, which is not indicative of the best accuracy level achievable in practice. With a larger user study we would expect to obtain better accuracy levels, matching those in Table~\ref{tab:mod_algo_results}, e.g., (TPR, TNR) $= (1.0, 0.6)$, and thus retrieving the secret in a smaller number of rounds.

\paragraph{Attack Performance without Timing Information}

It may appear from the lack of improvement in L2 and L3 adversaries' performance that the eye movement related features do not show any gain over simply timing based information. 
This is particularly problematic from an attacker's point-of-view as scheme designers can easily mask timing information by mandating a minimum time before the user can submit a response in each authentication round.
However, the eye movement features are also fairly accurate indicators of the modulus/no-modulus event. To demonstrate this, an experiment with the Total Time feature excluded from the feature set. With only a Naive Bayes classifier, we are able to obtain (TPR, TNR) = $(1.0, 0.38)$ with the second ranked feature (Dwell Consistency) at a threshold of 0.76. This offers performance equivalent to the adversaries L1-L3 in Table~\ref{tab:mod_classifier_algo_results}. This feature is part of the feature set of adversaries L2 to L4, and hence demonstrates that observing eye movement patterns can successfully retrieve the secret.

\paragraph{Per-User Accuracy Rate}
Until now we have reported system-wide accuracies to determine an attacker's performance. Since the dataset is small, we are interested in how the modulus detection accuracy varies between users, to see if the system-wide values are good representatives. We therefore report modulus detection accuracies for each user within our study for the four selected configurations (corresponding to adversary levels noted in Table~\ref{tab:mod_classifier_algo_results}). These are shown in Table~\ref{tab:per-user-acc}.

First, we see that for all users against all adversary levels, we achieve a TPR of 1.0. In case of TNR, against adversary levels L1-L3, 7 out of the 11 users are within $\pm 0.2$ of the system wide TNR of 0.38 (cf. Table~\ref{tab:mod_classifier_algo_results}). Three other users have TNRs between 0.6 and 0.667, slightly off the mark from the system TNR. One user, however, is an outlier with a TNR of 0.0. On the other hand, again, 7 out of 11 users against adversary level 4 are within $\pm 0.2$ of the system TNR of 0.4. However, the outliers in this case are further adrift, with 2 of the users exhibiting a TNR of 0.0, and 2 others showing a TNR of more than 0.714. Note, that higher than average TNR is not a problem from the attack's perspective, as this would require fewer observations before the secret can be retrieved (cf. Tables~\ref{tab:mod_algo_results_bc} and \ref{tab:mod_algo_results}). Thus, we can conclude that the system-wide performance of the attack is mostly representative of its performance per-user: the attack can be carried out against most users in the system, with TNR of most users being close to the system-wide TNR. This indicates that the classifiers are unlikely to have overfit. The exception being the outliers who exhibit a TNR of 0.0. The prevalence of such users in the general population would require a larger study, which we leave as future work. 

\begin{table}[t]
\caption{Modulus detection accuracy separated on a per-user basis. It is observed that in L1-3, the TNR is approximately equal between users. Under L4 however, there appears to be more variance in the performance of the classifier. Where no accuracy is reported for TPR, no positive user samples exist. The total number of user positive and negative samples are noted in the last row of the table.}
\label{tab:per-user-acc}
\vspace{-3mm}
\resizebox{1.0\columnwidth}{!}{%
\begin{tabular}{|l||cc|cc|cc|cc|cc|cc|cc|cc|cc|cc|cc|}
\hline
{User} & \multicolumn{2}{c|}{1} & \multicolumn{2}{c|}{2} & \multicolumn{2}{c|}{3} & \multicolumn{2}{c|}{4} & \multicolumn{2}{c|}{5} & \multicolumn{2}{c|}{6} & \multicolumn{2}{c|}{7} & \multicolumn{2}{c|}{8} & \multicolumn{2}{c|}{9} & \multicolumn{2}{c|}{10} & \multicolumn{2}{c|}{11} \\\cline{1-23}
 Adv. & TNR & TPR & TNR & TPR & TNR & TPR & TNR & TPR & TNR & TPR & TNR & TPR & TNR & TPR & TNR & TPR & TNR & TPR & TNR & TPR & TNR & TPR \\ \hline
L1 & 0.4 & 1 & 0.0 & 1 & 0.6 & 1 & 0.333 & 1 & 0.333 & - & 0.333 & 1 & 0.286 & - & 0.667 & - & 0.667 & 1 & 0.4 & 1 & 0.4 & 1 \\
L2 & 0.4 & 1 & 0.0 & 1 & 0.6 & 1 & 0.333 & 1 & 0.333 & - & 0.333 & 1 & 0.286 & - & 0.667 & - & 0.667 & 1 & 0.4 & 1 & 0.4 & 1 \\
L3 & 0.4 & 1 & 0.0 & 1 & 0.6 & 1 & 0.333 & 1 & 0.333 & - & 0.333 & 1 & 0.286 & - & 0.667 & - & 0.667 & 1 & 0.4 & 1 & 0.4 & 1 \\
L4 & 0.2 & 1 & 0.4 & 1 & 0.2 & 1 & 0.0 & 1 & 0.0 & - & 0.833 & 1 & 0.714 & - & 0.333 & - & 0.333 & 1 & 0.2 & 1 & 0.6 & 1 \\ \hline
Total & 5 & 2 & 5 & 1 & 5 & 1 & 3 & 3 & 3 & 0 & 6 & 1 & 7 & 0 & 3 & 0 & 3 & 2 & 5 & 3 & 5 & 1 \\ \hline
\end{tabular}
}
\vspace{3mm}
\end{table}

\section{Application to Other ORAS}
\label{sec:other_oras}
In this section, we show that the attack is applicable to other ORAS which do not fit the description of $k$-out-of-$n$ ORAS, as long as they contain a modulus operation. We use two such ORAS: PassGrids~\cite{kelleyimpact} and Mod10~\cite{wilfong1999method}, and present slightly modified point update algorithms tailored to these schemes. 
Both PassGrids and Mod10 use a modulus of $d=10$, and due to their fundamentally different construction from $k$-out-of-$n$ ORAS, not all side-channel features previously used are relevant (e.g., no items to gaze at in Mod10). Coarse timing information, however, is still relevant, due to the \emph{problem size effect} as studied by LeFevre, Sadesky and Bisanz.~\cite{lefevre1996selection}. 
The problem size effect observes relatively slower latency (timing) on arithmetic problems with sums greater than 10. For PassGrids and Mod10 with a modular operator of $d=10$, slower latency then is a close indicator of the modulus/no-modulus event. Thus, we may think of the faulty oracles in the attack algorithms on these schemes being initiated by classifiers that use such timing related information to classify modulus/no-modulus events. Throughout this section we will use symmetrical oracle accuracies despite our earlier observation of an asymmetrical response to classifier errors, this is to provide simpler performance references of hypothetical attackers.

These schemes can be configured with secrets of variable length. For example a PIN can be 4 or 6 digits in length. Each secret digit and the challenge cognitive function are independent of the other secret digits. As such, we assume our attacker is capable of obtaining oracle information for each sequential challenge (pass-item/digit), and has knowledge of when a challenge (pin digit entry) starts and stops. 
This notion was not applicable for the previous schemes of BehavioCog, FoxTail and HopperBlum, as the secret items collectively produce a single final response. We remark that while we do have oracle information about the individual digits, we do not stop updating points on any digit until all digits are ranked highest, i.e., the complete secret has been found. 


\subsection{PassGrids}

\begin{algorithm}[b]
\vfill
\SetAlgoLined
\DontPrintSemicolon
%
\KwIn{ Number of challenges $m$; A set of secrets $S$ where $s \in S$ is a tuple $(i, x, y)$, where $i$ is one of 36 locations, $x \in \{1, \ldots, 9\}$ and $y \in \{0, \ldots, 9\}$; size of $S$ as $n = 36 \times 9 \times 10$. }
\KwOut{A list of points $( p_1, p_2, \ldots, p_n )$, with top score indicating the target secret.}
Initialize $( p_1, p_2, \ldots, p_n)$ to all zeroes.\;
\For{$j = 1$ to $m$}{
        Observe challenge $c$, auxiliary information `aux,' and response $r$.\;
        $b \gets \mathcal{O}_{\text{mod}}^{\text{TPR}, \text{TNR}}(c, \text{aux})$.\;
        \If{$f(s,c) \ne r$, for $s \in S$}{
            penalize $s$ by 10\;
        }
        \Else{
            \If{$b = -1$ (no-modulus event) $\And{}$ $(f(s,c) \ge 10)$ \textbf{or} $b = +1$ (modulus event) $\And{}$ $(f(s,c) < 10)$}{
                penalize $s$ by 3\;
            }
    }
}
\Return{$( p_1, p_2, \ldots, p_n )$}.\; 
\caption{{\sc PassGrid Points Update}}
\label{algo:passgrid_update}
\end{algorithm}

The PassGrids system~\cite{kelleyimpact} consists of a series of schemes which are modifications of the commonplace PIN authentication systems. The schemes are designed to be resistant to observation. We consider the version of their scheme called ``PGx+4.'' This scheme is implemented on a $6\times6$ grid with 36 possible locations. A challenge consists of an assignment of a random digit $\{0, \ldots, 9\}$ to each of the 36 locations. The digits are generated so that each appears an approximately equal number of times, i.e., 3-4 times. The user's secret is a set of four tuples of the form: $(i, x_i, y_i)$, where $i$ is a random location, $x_i \in \{1, \ldots, 9\}$ and $y_i \in \{0, \ldots, 9\}$. For each secret tuple $s$, given the challenge $c$, the response is computed as $r_i = f(s, c) = c_ix_i + y_i \bmod 10$, where $c_i$ is the digit corresponding to location $i$ in the challenge. The secret space is thus of size $36\cdot 9 \cdot 10 = 3240$ for a 1-length secret, and consequently a 4-length secret would have $\frac{(3240)!}{(3240 - 4)!} \approx 2^{46.6}$ possible secrets. This scheme offers a degree of observation resilience ($<$10 observations). Once again, we see that the modulus operation is used to provide observation resilience.

We implemented the PGx+4 scheme, and ran our attack algorithm, Algorithm \ref{algo:passgrid_update}, on it. The attack algorithm is a point update algorithm that penalize locations and operands ($x_i$'s and $y_i$'s) that do no agree with the challenge response. We also to a lesser extent penalize secrets that do agree with the response but do not agree with the modulus oracle. The algorithm is shown for a 1-length secret for clarity, but can simply be extended to a 4-length secret by parallel execution.  

We selected a point update vector of 10 for secrets that conflict with the response and 3 for secrets that conflict with the modulus oracle. While we have chosen 10 and 3, as long as the first value is larger than the second, the performance will fall within the performance bounds from direct elimination of secrets (perfect knowledge).
Figure \ref{fig:crack_gridpass} displays a CDF on the percentage of 1000 PassGrids that is found after a given number of observations. This demonstrates the modulus information can be used to enhance an attack on this scheme.


\begin{figure}[t]
	\centering
    \includegraphics[width=0.6\columnwidth]{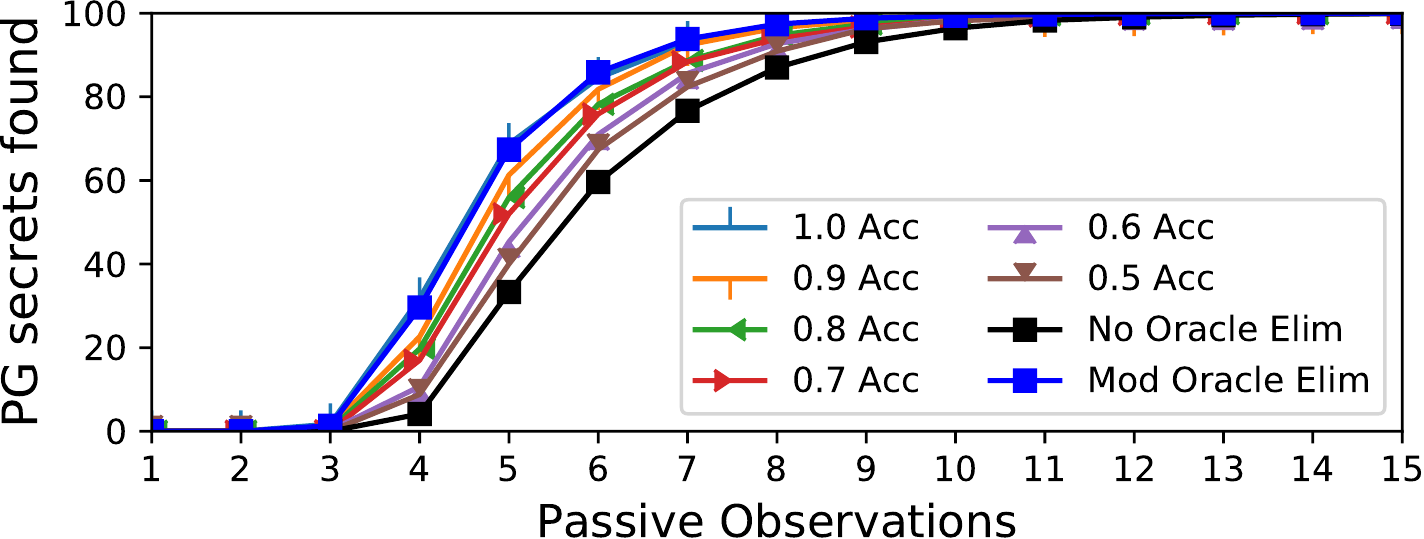}
    \vspace{-2mm}
  	\caption{The CDF of 1000 PassGrid user secrets found over a increasing number of observations attained by an attacker, with varying degrees of modulus information accuracy. We note that the square markers results eliminate possible secrets with perfect oracles, instead of updating points.}
 	\label{fig:crack_gridpass}
    \vspace{2mm}
\end{figure}


\subsection{Mod10}

\begin{algorithm}[b]
\SetAlgoLined
\DontPrintSemicolon

\KwIn{$m$ responses. }
\KwOut{A list of points $( p_0, p_1, \ldots, p_9 )$, with top score indicating the target secret digit.}
Initialize $( p_0, p_1, \ldots, p_9)$ to all zeroes.\;
\For{$j = 1$ to $m$}{
        Observe auxiliary information `aux,' and response $r$.\;
        $b \gets \mathcal{O}_{\text{mod}}^{\text{TPR}, \text{TNR}}(c, \text{aux})$.\;
        \If{$b = -1$ (no-modulus event)}{
            reward $(p_{0}, \ldots, p_r)$
        }
        \Else{
            reward $(p_{r+1}, \ldots, p_9)$
        }
}
\Return{$( p_0, p_1, \ldots, p_9)$}.\; 
\caption{{\sc Mod10 Points Update}}
\label{algo:mod10_update}
\end{algorithm}

\begin{table}[t!]
\caption{Response and Modulus operation (mod performed shaded) of a given secret digit and one time pad.}
\label{tab:mod10_bias}
\vspace{-2mm}
\centering
\resizebox{0.45\columnwidth}{!}{%
\begin{tabular}{| p{0.3cm}p{0.3cm} | llllllllll |}
\hline
 \multicolumn{2}{|c|}{Response}  & \multicolumn{10}{c|}{User Secret Digit} \\
 \multicolumn{2}{|c|}{of Sum}   & 0 & 1 & 2 & 3 & 4 & 5 & 6 & 7 & 8 & 9 \\ \hline
\multirow{10}{*}{\rotatebox[origin=c]{90}{One Time Pad digit}} & 0 & 0 & 1 & 2 & 3 & 4 & 5 & 6 & 7 & 8 & 9 \\
 & 1 & 1 & 2 & 3 & 4 & 5 & 6 & 7 & 8 & 9 & \cellcolor{green}0 \\
 & 2 & 2 & 3 & 4 & 5 & 6 & 7 & 8 & 9 & \cellcolor{green}0 & \cellcolor{green}1 \\
 & 3 & 3 & 4 & 5 & 6 & 7 & 8 & 9 & \cellcolor{green}0 & \cellcolor{green}1 & \cellcolor{green}2 \\
 & 4 & 4 & 5 & 6 & 7 & 8 & 9 & \cellcolor{green}0 & \cellcolor{green}1 & \cellcolor{green}2 & \cellcolor{green}3 \\
 & 5 & 5 & 6 & 7 & 8 & 9 & \cellcolor{green}0 & \cellcolor{green}1 & \cellcolor{green}2 & \cellcolor{green}3 & \cellcolor{green}4 \\
 & 6 & 6 & 7 & 8 & 9 & \cellcolor{green}0 & \cellcolor{green}1 & \cellcolor{green}2 & \cellcolor{green}3 & \cellcolor{green}4 & \cellcolor{green}5 \\
 & 7 & 7 & 8 & 9 & \cellcolor{green}0 & \cellcolor{green}1 & \cellcolor{green}2 & \cellcolor{green}3 & \cellcolor{green}4 & \cellcolor{green}5 & \cellcolor{green}6 \\
 & 8 & 8 & 9 & \cellcolor{green}0 & \cellcolor{green}1 & \cellcolor{green}2 & \cellcolor{green}3 & \cellcolor{green}4 & \cellcolor{green}5 & \cellcolor{green}6 & \cellcolor{green}7 \\
 & 9 & 9 & \cellcolor{green}0 & \cellcolor{green}1 & \cellcolor{green}2 & \cellcolor{green}3 & \cellcolor{green}4 & \cellcolor{green}5 & \cellcolor{green}6 & \cellcolor{green}7 & \cellcolor{green}8 \\ \hline
\end{tabular}
}
\vspace{2mm}
\end{table}

\begin{figure}[t!]
	\centering
    \includegraphics[width=0.6\columnwidth]{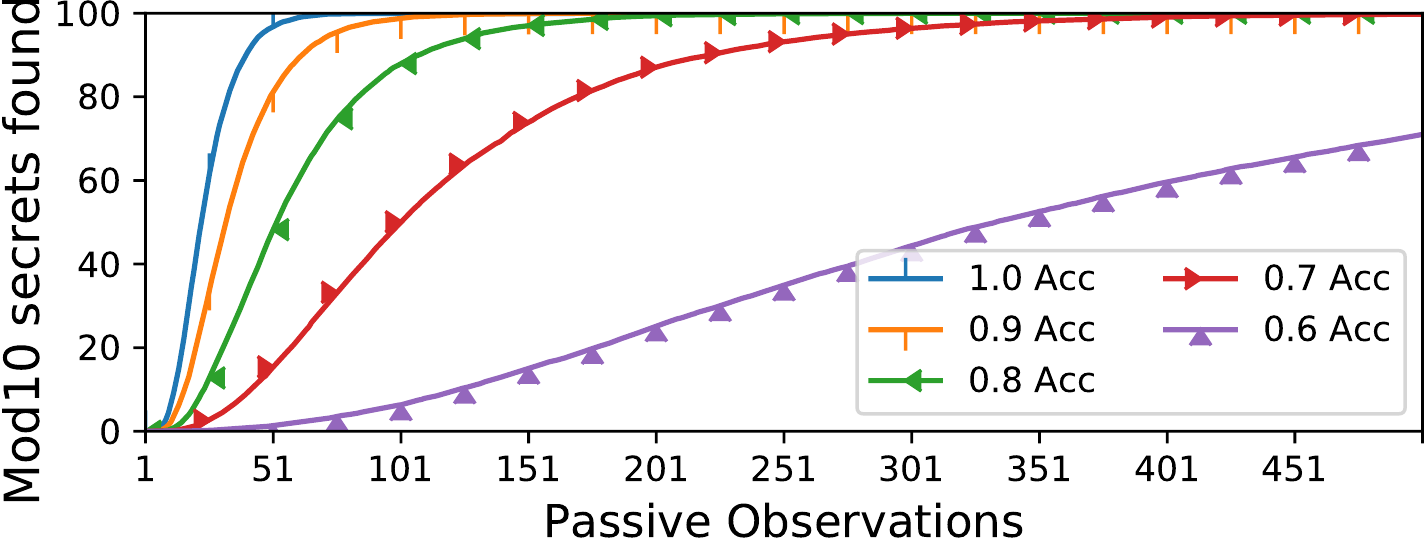}
    \vspace{-2mm}
  	\caption{The CDF of Mod10 user secrets found over a increasing number of observations attained by an attacker, with varying modulus oracle accuracy.}
 	\label{fig:crack_mod10}
 	\vspace{2mm}
\end{figure}

The Mod10 method~\cite{wilfong1999method, vcagalj2015timing} is a patented method proposed as an alternative to commonplace PIN authentication. The scheme combines each of the digits in the user's PIN with a one time pad (OTP) communicated through a protected channel. More specifically, for the $i$th PIN digit $s_i \in \{0, \ldots, 9\}$, the verifier covertly communicates an OTP $o_i \in \{0, \ldots, 9\}$. The user responds with $r_i = s_i + o_i \bmod 10$. It follows that each of the digits is equally likely to be the secret even after observing the response $r_i$.



We consider the effect of the faulty modulus oracle in obtaining the user's PIN. The knowledge of a modulus event divides the response space, as seen in Table \ref{tab:mod10_bias}. This knowledge reduces the number of possible secret pin digits which could be combined with an unknown random OTP to produce the given response. The user's 4-digit secret PIN can be found by updating points to reward secret digits that agree with the response, and the mod-oracle, for each of the 4 secret digits. The algorithm for one PIN digit is shown in Algorithm \ref{algo:mod10_update}, which can be extended in a modular way to all 4 digits. In the case of a no-modulus event, all secret digits less than or equal to the response $r$ are rewarded (since, digits greater than $r$ would require a modulus operation regardless of the OTP).
More precisely, since it is a no-modulus event, we necessarily have $s_i + o_i = r_i$ (even without reducing the result modulo 10). If $s_i > r_i$, then this implies $o_i < 0$, a contradiction. Hence, $s_i \le r_i$. Thus, we reward the points $p_0$ to $p_{r}$ in the algorithm. On the other hand, in case of the modulus event, all secret digits greater than the response $r_i$ are rewarded. This follows from the fact that in a modulus event, we necessarily have $o_i + s_i = r_i + 10$. Since, $o_i \le 9$, this gives us $r_i + 10 \le 9 + s_i$, and hence $s_i \ge r_i + 1$. Therefore, all points $p_{r+1}$ to $p_{9}$ are rewarded in the algorithm.
By simulating users on the Mod10 scheme, and using symmetrical oracle accuracies (same TPR and TNR for no-modulus/modulus events) we can find the PIN in
(Mod Accuracy, Average Rounds): (1.0, 24.4), (0.9, 36.1), (0.8, 60.20), (0.7, 118.37), (0.6, 409.76). This result is also visually displayed in Figure~\ref{fig:crack_mod10}. 

We note that Cagalj et al.~\cite{vcagalj2015timing} demonstrate timing attacks on the same scheme, exploiting the differences in the user's cognitive load in the addition of the one time pad to their secret digit. They are able to reduce the entropy of the unknown pin digit by 0.5 bits over an observation, and effectively reducing the candidate size of the pin digit from 10 to 6 (with 90\% confidence) over 90 observations. However, their attack does not retrieve the entire secret.





\section{Discussion}
\label{sec:discussion}
For $k$-out-of-$n$ ORAS, one way to reduce the efficacy of the proposed side-channel attack is to increase the expected number of secrets present in the challenge as shown in Section~\ref{sec:weight_bias}. Unfortunately, this comes at a substantial cost to usability, since the only way to increase the number of secrets (without introducing statistical vulnerabilities~\cite{yan2012limitations}) is to either increase the number of secret items $k$ or the window size $l$. Both increase the cognitive load on the user by requiring more secrets to recognize and more computations to perform. Decreasing the number of total items $n$ is not desirable either as it reduces the password space making it susceptible to brute-force attack.

As some of our features are directly related to timing information (e.g., total time of authentication), enforcing a minimum time before response submission, as suggested previously~\cite{vcagalj2015timing}, may increase the difficulty of detecting the modulus event. This mitigation technique is also applicable to the PassGrids and Mod10 schemes (which are not $k$-out-of-$n$ ORAS). However, this is ineffective against some of the other features used in our attack to detect the module event, e.g., dwell consistency. Thus, eye movement patterns can still reveal side channel information even with this mitigation technique. 

Since we only had a low number of user samples available, we restricted ourselves to simple classifiers to produce a fair model. We acknowledge that additional user samples would allow more sophisticated classifiers (e.g. neural networks) to be trained, producing an improved side-channel classifier, and thus faster compromise of the user secret. Additionally a more representative sample of the population may yield more diverse results in both timing and eye-tracking based features, as a majority of our participants were young, research students highly capable of doing basic mental mathematics.

While our attacks do not deem the ORAS considered in this paper completely insecure, they show that the security of these schemes is greatly reduced under side-channel attacks in terms of number of rounds a secret can be used before renewal. We have argued that a scheme's claim to being observation resilient should be evaluated against side-channels attacks as well. In the case of BehavioCog and FoxTail, we have less than halved the ``safe'' number of rounds for these schemes. On other schemes, such as the HB protocol and Mod10 our attacks have shown a stronger result. The former does not have an efficient algebraic/statistical attack, and the latter by definition is immune to any algebraic attacks (being an OTP-based scheme).

For a number of reasons, we did not pursue more intuitive attacks based on eye movement patterns, such as following the user's gaze and directly labeling items with higher dwell times as possible secret items. The main difficulty is in simulating a given accuracy level of an oracle which predicts an item being a secret or a decoy item. To simulate a given oracle accuracy level, we would need to first determine how it translates to the empty challenge event, i.e., when there are no secret items present in a challenge. Moreover, if it is not an empty challenge event, we need to determine how the oracle accuracy relates to the number of secret items present in the challenge, the dwell times for each of the secret/decoy items in the challenge, and the fact that the user sometimes does not dwell on the secret items at all (we found through our user study that some users would never dwell over any particular item, secret or decoy, when computing the response). Thus, while we could assume an oracle which predicts each of the $l$ items present in the challenge as being a secret/decoy item with a given accuracy level, such an oracle would be relying on a lot more assumptions which need to be justified. In comparison, the modulus event is a single binary event tied to the entire challenge. A larger user study would indicate whether such direct eye-gaze attacks are viable or not, by relying on empirical data rather than simulated oracles. Due to our limited user study, we were not able to do so.

Lastly, we remark that another advantage of the modulus-based attack over direct eye-gaze attacks, is that it delineates the attack algorithm from the actual side-channel being used. Eye movement patterns may not be the only source of side channel information. It could be possible that third party trackers on a device with access to any one of the many device sensors, may utilize this data and establish additional side-channel to expose the user's secret. This is an interesting avenue for future work.

\section{Related Work}
\label{sec:related}
We focus on related work on side-channel attacks on ORAS as well as password and PIN authentication schemes. The most related work to ours is the timing attack from Cagalj et al.~\cite{vcagalj2015timing} who exploit coarse-grained timing information as a side channel. Coarse-grained means that the timing information is limited to the overall time taken to respond to a challenge. They exploit the fact that the users time to respond to a challenge is directly proportional to the cognitive load (which varies due to randomized weights in the challenge). They demonstrate the susceptibility of the (full) HB protocol and the Mod10 scheme to their timing attack.
In contrast, our work exploits further information (features other than total time taken) obtainable via observing user's eye movement patterns coupled with the observation that a modulus event indicates a high cognitive-load challenge. As a result, our attack is applicable to a broader class of $k$-out-of-$n$ ORAS (as well as other ORAS that use a modulus operation). Note that the simple timing attack from~\cite{vcagalj2015timing} does not apply to the windowed HB protocol considered in this paper.

To the best of our knowledge, this is the only work that explores side-channel attacks on ORAS. However, there are numerous studies on side-channel attacks on PIN and password-based schemes, which we summarize next. 

Kune and Kim present a side-channel attack which extracts the user's PIN by observing the time taken as the finger travels between PIN digits on the keypad \cite{FooKune2010}. This timing information enables the attacker to derive the distance travelled, and thus infer potential key pairs the user was moving between. With the key pairs the attacker reduces the possible space of secrets, to eventually find the user's secret.
The attack from~\cite{cai2011touchlogger} uses the position of the phone during PIN entry to determine the location of the secret digit using gyroscope information. The smudge attack \cite{aviv2010smudge} is able to infer a user's pass-pattern from the oily residue remaining on the screen from the user's finger when in contact with the screen. We note that the challenge-response pairs in the authentication systems considered in this paper are already assumed to be known to the attacker, and as a result these side channel attacks are not applicable to our case. 

There has also been some work on using hidden and/or on-device cameras to steal user's PIN entry. The work in \cite{yue2014blind} shows how an attacker can use computer vision to determine the exact digit or keyboard letter pressed through a distant camera even if the angle is not optimal (directly facing the screen). Likewise, \cite{simon2013pin} shows a similar attacker capability who has access to a front facing camera feed of the user, to identify which digit was pressed on an on-screen number pad. They observe that the user (in one-handed operation) may tilt the phone, and consequently the camera to press a digit in a particular location. Thus if the location can be derived from the position of a stationary reference (e.g. the user's face) on the camera feed, so can the secret digit. Our work relates to these two works in terms of using a camera recording to detect eye movement patterns; however, as discussed before, the task of retrieving the secret in our case is more involved (as opposed to mere detection of password letters entered).


\section{Conclusion}
We have investigated and successfully exploited the modulus event present in existing observation resilient schemes. We have shown that proposed schemes are vulnerable to eye-movement based side-channel information which indicates the occurrence of the mod events. In our works we have presented the algorithms to exploit the weight bias in the modulus event, with an attempt at leveraging timing and positional focus found in a user's unconscious behavior in solving the authentication challenges. With the algorithm independent to the side-channel, we speculate there may exist other behavioral features that can be measured and utilized to better improve the overall attack. 
The development of algorithms to exploit cognitive schemes that involve the modulus like PassGrids or Mod10 demonstrate the value of this leaked information.
In this work, through analysis, we are able to derive why these side-channels leak information about the secret, present remedies to reduce the amount of information released in observation-resilient authentication schemes, and serves to inform future scheme designers.

\section{Acknowledgments}
This research was funded by the Optus Macquarie University Cybersecurity Hub, Data61 CSIRO, an Australian Government Research Training Program (RTP) Scholarship and the COMMANDO-HUMANS project (EPSRC Grant EP/N020111/1). We would like to thank the anonymous reviewers for their feedback to improve the paper.

\bibliographystyle{ACM-Reference-Format}
\bibliography{ref.bib}
\appendix

%
\section{Proof of Theorem \ref{theorem:prob_mod_1}}
\label{proof:prob_mod_1}
\begin{proof}
\begin{align}
\Pr(Y<d \mid g) &= \sum_{y=0}^{d-1} \Pr(Y=y \mid g), \nonumber \\
&= \frac{1}{d^g} \sum_{y=0}^{d-1} \sum_{s=0}^{\lfloor \sfrac{y}{d} \rfloor} (-1)^s \binom{g}{s} \binom{y - sd + g -1}{g-1}, \nonumber
\end{align}
Since $0 \le y < d$, we have $\lfloor \sfrac{y}{d} \rfloor = 0$. Thus, $s = 0$, and we get
\begin{align}
\Pr(Y<d \mid g) &= \frac{1}{d^g} \sum_{y=0}^{d-1} \binom{y + g -1}{g-1}, \nonumber \\
&= \frac{1}{d^g} \binom{0 + g -1}{g-1} + \frac{1}{d^g} \binom{1 + g -1}{g-1} + \cdots \nonumber\\ 
    &+ \frac{1}{d^g} \binom{d - 1 + g -1}{g-1} \nonumber \\
%
&= \frac{1}{0! d^g} + \frac{g}{1! d^g} + \frac{g(g+1)}{2! d^g} + \cdots \nonumber\\
&+ \frac{g(g+1)\cdots(g+d+1)}{(d-1)!d^g}
\nonumber
\end{align}
%
As $g \rightarrow \infty$, we see that each polynomial numerator is $o(d^g)$. Thus, $\Pr(Y < d \mid g) \to 0$. Or equivalently, $\Pr(Y\ge d \mid g) \to 1$.
\end{proof}

\thispagestyle{empty}
\section{Proof of Lemma~\ref{lem:d-ineq}}
\label{sec:lem1_proof}
\begin{proof}
First let $g = 2$.
The proof is by induction on $d \ge 1$. First let $d = 1$. Then since $p(0) > 0$,
\begin{align*}
\frac{1}{2}\sum_{i=0}^1 i p(i) &= \frac{1}{2}\cdot 0 \cdot p(0) + \frac{1}{2}\cdot 1 \cdot p(1) \\
                              &< \frac{1}{2}\cdot 1 \cdot p(0) + \frac{1}{2}\cdot 1 \cdot p(1). \\
                              &= \frac{1}{2}\sum_{i=0}^1 p(i).
\end{align*}
Thus, the statement is true for $d = 1$. Now assume the statement holds for $d = r$, then
\begin{align*}
\frac{1}{2}\sum_{i=0}^{r + 1} i p(i) &=  \frac{1}{2}\sum_{i=0}^{r} i p(i) + \frac{r + 1}{2}p(r+1)\\
                                     &< \frac{r}{2}\sum_{i=0}^{r} p(i) + \frac{r + 1}{2}p(r+1) \\
                                     &< \frac{r+1}{2}\sum_{i=0}^{r+1} p(i), \\
\end{align*}
which completes the proof for $g = 2$. For $g > 2$, observe that
\begin{align*}
\frac{1}{g} \sum_{i=0}^{d} i p(i) < \frac{1}{2} \sum_{i=0}^{d} i p(i),
\end{align*}
and hence the lemma is true for all $g \ge 2$.
\end{proof}

\section{Proof of Theorem \ref{the:no-mod-point-update}}
\label{proof:no-mod-point-update}
\begin{proof}
Let $i \in [n]$ be a secret item and let $j \in [n], j \neq i$ be a decoy item. Let $\eta(i)$ and $\eta(j)$ denote the number of times the two items appear in $m$ challenges. Let $\eta^+(i)$ and $\eta^-(i)$ denote the number of times the secret item $i$ appears in the modulus and no-modulus events, respectively. First, for sufficiently large $m$, we see that both $\eta(i)$ and $\eta(j)$ approach their expected value, and therefore
\begin{equation*}
 \eta(i) \approx \eta(j) = \eta^+(j) + \eta^-(j). 
\end{equation*}
Next, note that due to step 6 in the algorithm, the secret item never gets penalized in case the oracle correctly identifies the no-modulus event (the secret item if present cannot have weight more than the response $r$). Therefore, we are looking at the instances where the oracle wrongly labels a modulus challenge as a no-modulus challenge. The probability of a particular response in this case is $1/d$. Since the secret item's weight is randomly generated, the probability that its weight is greater than $r = i$ is given by $(d - 1 - i)/d$. Therefore, the expected points update is given by
\[
\frac{1}{d^2} \sum_{i = 0}^{d - 1} (d - 1 - i)u_i.
\]
Denote the above by $u$. Then, the expected score of a secret item $i$ in $m$ challenges is given by

\begin{align*}
    (1 - \text{TPR}) \cdot \eta^+(i) \cdot u & \le  (1 - \text{TPR}) \cdot \eta(i) \cdot u \\
    &\approx (1 - \text{TPR}) \cdot \eta(j) \cdot u \\
    &= (1 - \text{TPR}) \cdot \eta^-(j) \cdot u \\
    &+ (1 - \text{TPR}) \cdot \eta^+(j) \cdot u \\
    &< \text{TNR} \cdot \eta^-(j) \cdot u \\
    &+ (1 - \text{TPR}) \cdot \eta^+(j) \cdot u,
\end{align*}

which is the expected score of the decoy item $j$ in $m$ challenges.
\end{proof}


\section{Feature Intuition}
\label{appendix:features}

Recall that the Dwell is period of user visual intake of a specific item, characterized by lowered rapid eye movement.

\subsection{Adversary Level 1 Feature Hypotheses}\label{tab:l1}
\begin{enumerate}[label=\alph*)]
\item \textbf{Total Time:} A challenge requiring a modulus operation involves more mental operations (size-effect-problem \cite{lefevre1996selection}), and should require more time.

\item \textbf{Mean Challenge Weight:} The expected value of individual weights is uniform, however collectively challenge may have a bias in the item weights. E.g. there are more higher weights, potentially providing information about the modulus event.

\item \textbf{Challenge Response:} As previously noted, there exists a small bias in the probability of a modulus occurring dependent on the final submitted response. This may be useful for informing the classifier.
\end{enumerate}

\subsection{Adversary Level 2 Feature Hypotheses}
\label{tab:l2}
\begin{enumerate}[label=\alph*)]
\item \textbf{Min Dwell Time:} The shortest time spent viewing an image can be indicative of the user's confidence that a secret image has been located. This value should be shorter when secret items are present. Alternatively this value may be short for when a user retrieves weights from a low number of secret for mental computation (e.g. 1 secret requires no computation).
However, a user quickly scanning in the challenge may also exhibit a short min time, which can be managed by considering the 20\textsuperscript{th} percentile.

\item \textbf{Max Dwell Time:} The longest time spent viewing an image may be indicative of the time that a user spends stationary to compute the challenge result. A more difficult arithmetic problem should incite a larger cognitive load and hence require more time. Like min dwell time, the longest dwell may reflect instances of user distraction for an extended period of time, hence the consideration of the 80\textsuperscript{th} percentile of dwell times.

\item \textbf{Mean Dwell Time:} If there are more secret images, with more math, the user may spend more time processing the challenge (Feature 1.a). But, the verification time of each image may be shorter as they skim over the challenge once again to retrieve item weights for mental computation.

\item \textbf{STD Dwell Time:} It is observed that users are more likely to double check the challenge if a low number secrets are present. Spending more time on specific uncertain images, should result in a larger time deviation.

\item \textbf{Number of Dwells:} The number of dwell positions should be indicative of the extent of scans and checks for secrets in the challenge. A challenge with more secrets may prompt additional checks, producing more dwells.

\item \textbf{Time from longest stationary till end:} A challenge with a secret present, should have the user stop and (mentally) compute a result. After the pause, they will submit their response. This is an attempt to isolate the period of time in which the user should be computing their response, and indicative of the problem difficulty.

\item \textbf{Dwell Consistency:} By contrasting the high and low extremes (min/max or 20th/80th percentile), we can obtain a normalized ratio of their differences. Thus, any outlying images such as a secret the user spends additional time on, will be captured by this normalized difference.

\item \textbf{Duration of First Fixation:} If a user locates a secret image initially when the challenge is presented, they will remain fixated on their secret image for a longer duration of time as compared to decoys \cite{elazary2010bayesian}.

\item \textbf{Duration of Last Fixation:} When a user recovers their secret items from the challenge, their last fixation would also include computation time for the modulus-sum obtain the result. The length of this last fixation is a possible indicator of the difficulty of the computational task, with the inclusion of the modulus hypothesized to take longer.

\item \textbf{Longest Dwell Consistency:} As an extension of the previous point, consistently taking a long time traversing multiple items may be indicative of a difficult task like that of the modulus.  
\end{enumerate}

\subsection{Adversary Level 3 Feature Hypotheses}
\label{tab:l3}
\begin{enumerate}[label=\alph*)]
\item \textbf{Number of transitions (Halves)} A user scanning through a challenge is likely to traverse the entire challenge, consequently crossing between different areas of the challenge. It is suspected to be larger for challenges with more secrets present. Transitions include: Left-Right (Ignoring center due to odd \# of columns), Top-Bottom and (Even), ignoring two center rows.

\item \textbf{Time from bottom of the screen to the end:} After a user finds their secrets, they press a button to proceed to a submission page, a user may take additional time to (mentally) compute the response prior to proceeding in a modulus event with many secrets.
\end{enumerate}

\subsection{Adversary Level 4 Feature Hypotheses}
\label{tab:l4}
\begin{enumerate}[label=\alph*)]
\item \textbf{Highest Number of Reentries:} For a given secret image within the challenge, it is likely the user's first pass will view the image to simply recognize it. However, upon completion of a visual search, the user may revisit the image to get the weight for response computation. Potentially, leading to larger values when many secrets are present.

\item \textbf{Number of non-entries:} For a given challenge, a user may quickly re-identify their secrets from a rapid search (no dwell), the secrets form salient images. As such, some images may not be viewed at all, thus producing more un-viewed images when less secret images are present.

\item \textbf{Length of longest repeating sequence:} During the visual search, a user is may backtrack on the items identified as secrets, either from uncertainty, or a revisit to retrieve weights for response computation. Thus a longer repeating sequence could be related with a larger number of secret items in a challenge, and thus provide modulus event information.

\item \textbf{Weight of Longest dwell item (Top 3):} As previously mentioned a user spends more time on secret items. Therefore larger weights on these dwelled items will likely require a modulus operation. The weights of the top 3 largest dwelled items are considered.
\end{enumerate}

\end{document}